\documentclass[11pt, letterpaper]{amsart}
\usepackage{color}
\usepackage{enumerate}
\usepackage{amssymb}
\usepackage{epsfig,psfrag,latexsym}

\newtheorem{thm}{Theorem}[section]
\newtheorem{cor}[thm]{Corollary}
\newtheorem{lem}[thm]{Lemma}
\newtheorem{prop}[thm]{Proposition}
\theoremstyle{definition}
\newtheorem{defn}[thm]{Definition}

\newtheorem{ass}[thm]{Assumption}

\theoremstyle{remark}
\newtheorem{rem}[thm]{Remark}
\newtheorem{exa}[thm]{Example}
\numberwithin{equation}{section}

\newcommand{\eps}{\varepsilon}

\newcommand{\A}{\mathcal{A}}
\newcommand{\K}{\mathcal{K}}
\newcommand{\such}{\ | \ }

\newcommand{\prob}{\mathbb{P}}

\newcommand{\qprob}{\mathbb{Q}}
\newcommand{\argmax}{\operatorname{argmax}}

\newcommand{\espalt}[3]{\mathbb{E}^{#1}_{#2}\left[#3\right]}
\newcommand{\tespalt}[3]{\wt{\mathbb{E}}^{#1}_{#2}\left[#3\right]}

\newcommand{\var}[1]{\mathrm{Var}\left[#1\right]}
\newcommand{\covar}[2]{\mathrm{Cov}\left[#1,#2\right]}
\newcommand{\probtriple}{(\Omega, \mathcal{F}, \prob)}

\newcommand{\filt}{\mathbb{F}}

\newcommand{\F}{\mathcal{F}}

\newcommand{\E}{\mathcal{E}}

\newcommand{\Leb}{\mathsf{Leb}}

\newcommand{\qtau}{\mathfrak{u}}

\newcommand{\nada}[1]{}

\newcommand{\dfn}{\, := \,}

\newcommand{\bra}[1]{\left[#1\right]}
\newcommand{\cbra}[1]{\left\{#1\right\}}
\newcommand{\dbra}[1]{[\kern-0.15em[ #1 ]\kern-0.15em]}
\newcommand{\dbraco}[1]{[\kern-0.15em[ #1 [\kern-0.15em[}

\newcommand{\tprob}{\widetilde{\mathbb{P}}}
\newcommand{\wt}[1]{\widetilde{#1}}

\newcommand{\reals}{\ensuremath{\mathbb R}}

\newcommand{\mcp}{\mathcal{P}}
\newcommand{\esp}{\mathbb{E}}
\newcommand{\condespalt}[3]{\esp^{#1}\bra{#2\big| #3}}
\newcommand{\ol}[1]{\overline{#1}}
\newcommand{\ul}[1]{\underline{#1}}

\newcommand{\essinf}[2]{\textrm{essinf}_{#1}\left(#2\right)}
\newcommand{\esssup}[2]{\textrm{esssup}_{#1}\left(#2\right)}
\newcommand{\leb}{\textrm{Leb}}
\newcommand{\hphi}{\hat{\varphi}}

\title[Optimal Investment, Demand and Arbitrage under Price Impact ]{Optimal Investment, Derivative Demand and Arbitrage under Price Impact}

\author{Michail Anthropelos}
\address{Department of Banking and Financial Management\\
University of Piraeus\\
Piraeus, Greece}
\email{anthropel@unipi.gr}
\thanks{M. Anthropelos is supported in part by the Research Center of the University of Piraeus}

\author{Scott Robertson}
\address{Questrom School of Business\\
Boston University\\
Boston, MA 02215}
\email{scottrob@bu.edu}
\thanks{S. Robertson is supported in part by the National Science Foundation
  under grant number DMS-1613159.}

\author{Konstantinos Spiliopoulos}
\address{Department of Mathematics \& Statistics\\
Boston University\\
Boston, MA 02215}
\email{kspiliop@math.bu.edu}
\thanks{K. Spiliopoulos is supported in part by the National Science Foundation under grant number DMS-1550918. }

\date{\today}

\begin{document}

\begin{abstract}
This paper studies the optimal investment problem with random endowment in an inventory-based price impact model with competitive market makers.  Our goal is to analyze how price impact affects optimal policies, as well as both pricing rules and demand schedules for contingent claims. For exponential market makers preferences, we establish two effects due to price impact: constrained trading, and non-linear hedging costs. To the former, wealth processes in the impact model are identified with those in a model without impact, but with constrained trading, where the (random) constraint set is generically neither closed nor convex. Regarding hedging, non-linear hedging costs motivate the study of arbitrage free prices for the claim. We provide three such notions, which coincide in the frictionless case, but which dramatically differ in the presence of price impact. Additionally, we show arbitrage opportunities, should they arise from claim prices, can be exploited only for limited position sizes, and may be ignored if outweighed by hedging considerations. We also show that arbitrage inducing prices may arise endogenously in equilibrium, and that equilibrium positions are inversely proportional to the market makers' representative risk aversion. Therefore, large positions endogenously arise in the limit of either market maker risk neutrality, or a large number of market makers.
\end{abstract}

\maketitle
\section{Introduction}\label{S:intro}

In this paper, we study the effects of inventory-based price impact on optimal investment, as well as the pricing of, and optimal demand for, contingent claims. In our model, price impact occurs endogenously when the investor submits an order to the market makers, as it changes their aggregate inventory and hence pricing rule. Our goal is to understand how price impact in the underlying market affects optimal investment in both the underlying assets and the claim. To this end, we establish two primary effects due to price impact: constrained trading and non-linear hedging costs.

Non-linear hedging costs mean that, even if the constraint is absent and the investor can fully hedge any position on the claim, the cost of the hedging strategy is not linear with respect to the position size, as it is in the frictionless case. This is due to price impact, and requires us to both revisit the notion of arbitrage-free prices and ask whether it is always optimal to exploit an arbitrage opportunity, since price impact in the underlying market may limit the arbitrage's gains. Furthermore, non-linear costs naturally lead us to consider the investor's demand schedule (i.e.~optimal demand for a given traded price) and ask whether it can ever be optimal to take a large position in the claim. Motivated by \cite{MR3678485}, where for models without price impact, large positions endogenously arise in conjunction with either vanishing hedging errors, risk aversion or transactions costs, we want to examine whether optimal demands become large as the price impact vanishes.

Rather than specifying the price impact of trades exogenously, as in \cite{MR2030833,MR1398050, MR1482708, jarrow1994derivative, sircar1998general, MR1613291, MR1776395,RocSon13} amongst others, we follow the financial economics paradigm, where the impact is determined endogenously. We use the impact models first developed in \cite{Ho198147,stoll1978}, and subsequently extended by \cite{MR3320327,MR3375887}. Here, market makers compete when quoting prices for a given demand process. The competition causes market makers to form Pareto optimal allocations, and cash balances to preserve their expected utilities. The cash balance process then determines a wealth process for the large investor. The problem of hedging claims in this setting was considered in the recent works \cite{MR3778361, sah2014hedging}.  Here, as a natural counterpart, we consider the dynamic optimal investment problem with random endowments (the corresponding static problem is studied in \cite{AnthBanGok18}, while \cite{fukasawa18} contains partial results in continuous time).

The first step is to examine the wealth process associated to a given investor trading strategy. We show (as was done in \cite{said2007} for a single asset) that under exponential preferences for the market makers, wealth processes in the price impact model can be identified with those in fictitious market without price impact, but with constrained trading. However, the constraint set varies with both time and scenario, and as highlighted by examples in Section \ref{SS:cons_set}, typically it is neither closed nor convex. As such, much of the standard theory for solving the investor's portfolio choice problem does not apply (c.f.~\cite{cvitanic1992cdc}). Adding to the difficulty is that the (if finite) boundary of the constraint set may correspond to an infinite investor demand in the price impact model, and at least formally, infinite demands may lead to well defined wealth processes (see Section \ref{S:cara_li} for a more detailed discussion). In order to rule out this situation, care must be taken when selecting the terminal payoff of the traded assets, and its relationship to aggregate market maker endowment.

While general conclusions regarding optimal investment are lacking due to the nature of the constraint set, it does turn out that in many cases of interest the constraint set is non-binding for the optimal investment problem.  For example, in the Bachelier model the constraint set is absent. Additionally, in the general case, if the market makers' and investor's endowments are  ``securitized'' (portfolios of the traded assets plus an independent component) optimal policies are static, and hence always admissible. Indeed, absent any exogenous order flow to the market makers (by other investors or ``noise traders''), the initial optimal order puts the large investor and the market makers to a Pareto-optimal situation, and no further trading is necessary. We establish these results by linking CARA investor preferences in the price impact model with CRRA preferences in the constrained trading model.

Allowing for the above conclusions, our main focus is dedicated to identifying pricing rules and demand schedules when price impact is present in the underlying market. We observe that hedging a position on a claim yields two separate but related questions: can we hedge? and how much does it cost to hedge?  To the second question we also ask if hedging costs are linear in the position size, and if not, what are the consequences?

To isolate the effects of hedging costs upon optimal demands, we remove the constrained trading effect, assuming for any position in the claim that a replicating strategy for that position exists.  Here, we show the hedging cost is not linear in the units of the claim.  As impact enlists doubt as to the scalability of arbitrage opportunities, non-linear hedging costs mandate we revisit the notion of arbitrage-free pricing. For this, we provide three definitions of an arbitrage-free price, which coincide in the frictionless case, but which in this price impact model, yield strikingly different arbitrage-free price ranges. Indeed, we obtain arbitrage-free price ranges which vary from a singleton all the way to the maximal interval, determined by claim's essential infimum and essential supremum.

While by definition a non-arbitrage free price gives rise to an arbitrage opportunity, due to price impact, the gain from the arbitrage is \textit{limited} in that it can be exploited only for small (enough) positions in the claim: for large positions the arbitrage vanishes. This follows because when the investor trades the arbitrage opportunity with the market makers, she changes their inventory and hence pricing rules. After a certain threshold in the claim's position size, hedging becomes  expensive and the arbitrage vanishes. Therefore, if the investor wants to buy/sell the derivative for hedging reasons, she may ignore an arbitrage opportunity when the benefits from hedging outweigh the limited gain from arbitrage. While this is in \textit{sharp contrast} with the notion of arbitrage in frictionless markets without price impact, it is consistent with the ``limits to arbitrage'' evidence in markets with frictions such as capital constraints or transaction costs (see \cite{AchLochRam13} for related empirical results). Similar non-scalable arbitrage opportunities also appear in \cite{GuaRas15} for continuous-time models with general market frictions. Lastly we comment that while pricing and demand research for contingent claims is well studied in frictionless markets, corresponding studies in the presence of price impact are scarce (for exceptions, see \cite{LiuYong05, MR3519167}).

The above implies optimal demand for a claim may be finite even if the traded price allows for arbitrage. In fact, optimal demands also depend on investor endowment and its relationship with both the claim and underlying assets. But, when is the optimal demand large? To answer this question we apply the results of \cite{MR3678485} and show optimal demands are inversely proportional to the (representative) market maker risk aversion. Thus, positions become large as either a \emph{single} market maker becomes risk neutral (a common assumption in microeconomics, c.f.~\cite{KylObiWa18}) or as the number of market makers increases (since the representative risk tolerance aggregates individual risk tolerances). As such, our results support the observed position sizes of over-the-counter (OTC) derivatives, accounting for price impact, provided either a large number of market makers, or market makers sufficiently close to risk-neutrality.

Lastly, we ask whether, and under which circumstances, the traded price for the claim (arbitrage-free or not) arises endogenously in equilibrium. For this, we use the notion of Partial Equilibrium Price Quantities (PEPQ) as introduced in \cite{MR2667897}, where two large investors engage in an OTC transaction, at a price and quantity which is utility-optimal for each. However, as we deal with price impact, care must be taken when describing exactly whom the  investors are trading with in order to hedge their risks. We take the perspective of a \textit{segmented market}, which has recently attracted significant attention (c.f.~\cite{RahZig09,RahZig14}). Here, price impact is ``local'', in that each investor trades in her local market,  with (local) market maker and traded assets. Then, the investors additionally trade with each other on a separate claim, according to the PEPQ formulation. Segmentation avoids the investors from hedging demands with the same market maker (and hence having to declare a cash-balance sharing rule), and enables private trading of the derivative rather then embedding the derivative into the previously existing traded assets. Assuming CARA preferences for both investors, we provide sufficient conditions for the existence and the uniqueness of the PEPQ. Remarkably, when the investors' endowments are sufficiently different, the endogenously arising equilibrium prices may not be arbitrage-free. Now, equilibrium prices that lie outside of the range of arbitrage-free prices have appeared in the financial economics literature. For example, \cite{KouPap2004} shows that arbitrage opportunities are compatible with equilibrium in a strategic security market model, while in \cite{CorGop10}, arbitrage equilibrium prices occur under restricted participation and exogenous portfolio constraints. However, to the best of our knowledge, this is the first time equilibrium arbitrage prices appear in the context of contingent claim pricing.

The rest of the paper is organized as follows. In Section \ref{S:Model}, we describe the price impact model, connect it with a fictitious market model with random constraints, and discuss the properties of the constraint set. In Section \ref{S:cara_li}, we specify to when the investor has exponential preferences and provide examples where the constraint is non-binding. Section \ref{S:Pricing} is dedicated to contingent claim demand, pricing and arbitrage. In Section \ref{S:PEPQ} we endogenize the traded derivative prices and positions using the PEPQ framework in segmented markets.  Lastly, Appendix \ref{AS:Proofs} contains proofs of the main results.

\section{The price impact model and fictitious market}\label{S:Model}

\subsection{The model}\label{SS:model}

The inventory-based price impact model we consider was introduced in \cite{Ho198147, stoll1978} and further developed in \cite{MR3320327,MR3375887}. There is a complete probability space $\probtriple$ which supports a $d$-dimensional Brownian motion $B$, and we denote by $\filt$ the $\prob$-augmentation of the natural filtration for $B$, so $\filt$ satisfies the usual conditions of right continuity and saturation at $0$ of all $\prob$-null sets.

The time horizon is $T>0$, and there are a finite number of securities with terminal payoff denoted by the $k$-dimensional, $\F_T$ measurable random vector $\Psi$\footnote{Throughout, all random variables are assumed $\F_T$ measurable.}. Transactions on the securities are made through a collection of market makers, who quote demand-based prices. The market makers are risk averse, have random endowments, and for a given order $q\in\reals^k$, (aggregately) ask a price $X(q)\in\reals$ according to two conditions: (i) the total order, price and random endowment are distributed between market makers in a Pareto-optimal way; and (ii) each market maker remains at indifference (i.e.~her expected utility remains unchanged) after the transaction. Every time an order is satisfied by the market makers, their inventory and hence pricing rule changes. In continuous time, both the order flow $Q=\{Q_t\}_{t\leq T}$ and the cash balance $X(Q) = \{X_t(Q)\}_{t\leq T}$ are adapted stochastic processes.

We assume exponential (CARA) preferences for the market makers. This allows us to consider hereafter a single (representative) market maker. Indeed, it is well-known (see \cite{MR2211128, Bor62, Wil68}) that the Pareto-optimal sharing problem of exponential utility maximizers can be written as an optimization problem of a representative agent, whose utility is again exponential, and whose risk tolerance and endowment are, respectively, the sum of individual market makers' risk tolerances and endowments. Furthermore, condition (ii) above implies the cash balance asked by the market makers is the indifference price of the representative agent (c.f.~\cite[Theorem 3.1]{MR3320327}, \cite[Theorem 4.9]{MR3375887}).

\begin{rem}
As argued in \cite{stoll1978}, the assumption of market makers' indifference pricing is economically reasonable. It is justified not only by intense competition among market makers (see the related discussion in \cite{AnthBanGok18} and \cite{MR3320327}), but also by the well-documented fact that market makers are willing to offer better prices to large investors (see \cite{AitGarSwa95, GoldIrvKanWie09} and the references therein). In our setting, ``better price'' means each market maker asks for the minimum price which compensates him for the position (i.e.~her indifference value).
\end{rem}

The representative market maker has risk aversion $\gamma>0$ and endowment $\Sigma_0$. We make the standing assumption on $\Sigma_0$, $\gamma$ and $\Psi$ (c.f.~\cite[Assumption 2.4]{MR3375887}, \cite[equation (15)]{MR3005018}).
\begin{ass}\label{A: integrability}
For all $q>0$, $\espalt{}{}{e^{-\gamma \Sigma_0 + q|\Psi|}} < \infty$.
\end{ass}

Consider first a single period model. As in \cite{MR3320327}, we take the market maker's perspective, and assume an investor submits an order for $-q\in\reals^k$ units of $\Psi$.  By assumption, the market maker quotes a cash balance $X(q)$ to remain at indifference to the order.  Since her endowment changes from $\Sigma_0$ to $\Sigma_0 + X(q) + q'\Psi$, $X(q)$ must satisfy
\begin{equation*}
\espalt{}{}{-e^{-\gamma(\Sigma_0 + X(q) + q'\Psi)}} = \espalt{}{}{-e^{{-\gamma\Sigma_0}}}.
\end{equation*}
In other words, $X(q)$ is the (seller's) indifference price for $-q$ units of $\Psi$, under risk aversion $\gamma$ and endowment $\Sigma_0$. For the investor, the wealth associated to the order at the end of the period is $V(q) = -q'\Psi - X(q)$.

In continuous time (c.f.~\cite{MR3375887}), an order flow is submitted by the investor, in the form of an adapted stochastic process $Q$, and the cash balance $X(Q)$ is an induced stochastic process. Here, for each $t\in[0, T]$, $Q_t$ is the cumulative number of securities sold by the investor to the market maker, and $X_t(Q)$ is the cumulative cash balance. As in the static case, $X(Q)$ is determined through the principle of market maker indifference. For a given $Q$, we denote by $V(Q) = \{V_t(Q)\}_{t\leq T}$ the investor's gains process, in that $V_t(Q)$ is the cash amount the investor will receive if she sells her cumulative order (see \cite[Equation (4.19)]{MR3375887}). In analogy to static case, at the terminal time  $V_T(Q)= -Q_T'\Psi-X_T(Q)$.

To precisely identify $V(Q)$, we need to introduce some notation. In view of Assumption \ref{A: integrability}, for all $q\in\reals^k$ the process
\begin{equation*}
N_t(q) \dfn \condespalt{}{e^{-\gamma\Sigma_0 - \gamma q'\Psi}}{\F_t};\qquad t\leq T,
\end{equation*}
is a strictly positive martingale, and hence by predictable representation, we may write
\begin{equation}\label{E:Nq}
\frac{N_t(q)}{N_0(q)} =  \E\left(\int_0^\cdot H_s(q)'dB_s\right)_t;\qquad t\leq T,
\end{equation}
for some predictable process $H(q)$. The map $(t,\omega,q)\to H_t(q)(\omega)$ is $\mathcal{P}(\filt)\otimes \mathcal{B}(\reals^k)$ measurable, where $\mathcal{P}(\filt)$ is the $\filt$-predictable sigma-algebra, and in fact, $H(q)$ is regular in $q$ as described in \cite[Lemma 5.5]{MR3005018}, as well as in \cite{PulidoKramkov17}. Next, define the class of processes
\begin{equation*}
\mathcal{A}_{PI}\dfn \cbra{Q \in \mathcal{P}(\filt) \such \int_0^T |H_t(Q_t)|^2 dt < \infty \textrm{ a.s.}}.
\end{equation*}
Using \cite[Theorem 3.2, Section 5.2]{MR3005018}, \cite[Theorem 4.9]{MR3375887}, we obtain the following representation for $V(Q)$ in terms of $H(Q)$. The proof of Lemma \ref{L:PI_gains_process} is in Appendix \ref{AS:Proofs}.

\begin{lem}\label{L:PI_gains_process}
Let Assumption \ref{A: integrability} hold. Then, for $Q\in\A_{PI}$, $V(Q)$ is well defined with
\begin{equation}\label{E:V_good}
\begin{split}
V_t(Q) &=\frac{1}{\gamma}\int_0^t (H_s(Q_s)-H_s(0))'(dB_s-H_s(0)ds) - \frac{1}{2\gamma}\int_0^t\left|H_s(Q_s)-H_s(0)\right|^2ds;\quad t\leq T.
\end{split}
\end{equation}
\end{lem}

\subsection{Price impact and constrained investment}\label{SS:constrain}

We now show how wealth processes $V(Q)$ for $Q\in \A_{PI}$ are directly related to wealth processes $X(\pi)$ obtained in a fictitious model with no price impact, but where there is constrained trading.  This fact was first shown (for a single asset) in the PhD thesis \cite{said2007}. It was subsequently observed (also for one asset) in the PhD thesis \cite{sah2014hedging}, as well as the unpublished \cite{fukasawa18}. It was also implicitly alluded to, albeit in a much more abstract setting, in the recent  \cite{MR3778361} (c.f.~Condition 3.7). Here, we make the observation both more formal and general, by allowing for a vector of securities. We emphasize the discussion on the constraint set and exploit this representation to analyze the problem of contingent claim pricing and hedging under price impact.

Consider a fictitious market where the money market account is set to $1$. There are $d$ tradeable assets with an exogenously given price process $S$ evolving according to
\begin{equation}\label{E:S_dynamics}
\frac{dS_t}{S_t} = \lambda_tdt+dB_t;\qquad t\leq T,
\end{equation}
for a (to-be-determined) predictable $d$-dimensional process $\lambda$, such that $\int_0^T |\lambda_t|^2 dt < \infty$, a.s.. Provided the requisite integrability, $\lambda$ is the unique market price of risk and, by construction, there is a unique measure $\qprob_0$ on $\F_T$ under which $S$ is a true martingale. $\qprob_0$ has density
\begin{equation}\label{E:rn_deriv_lambda}
\left.\frac{d\qprob_0}{d\prob}\right\vert_{\F_T}= \mathcal{E}\left(-\int_0^\cdot \lambda_t'dB_t\right)_T.
\end{equation}
Self-financing trading strategies are denoted by $\pi$, where $\pi^{i}_t$ is the \textit{proportion} of wealth invested in $S^i$ at $t$, $i=1,...,d$. We assume $\pi\in\mcp(\filt)$ is such that $\int_0^T |\pi_t^2|dt < \infty$ a.s..  The wealth process induced by $\pi$ has dynamics
\begin{equation*}
\frac{dX_t(\pi)}{X_t(\pi)} = \pi_t'\left(\lambda_tdt+dB_t\right);\qquad t\leq T,
\end{equation*}
so that, with initial wealth $X_0 = e^{\gamma x}$, the log wealth process is
\begin{equation}\label{E:cons_wealth}
\log\left(X_t(\pi)\right) = \gamma x+ \int_0^t \pi_t'(dB_t+\lambda_tdt) - \frac{1}{2}\int_0^t |\pi_t|^2dt;\qquad t\leq T.
\end{equation}
Now, if we define $\lambda\dfn -H(0)$, assume $\pi = H(Q)-H(0)$, and compare \eqref{E:cons_wealth} with \eqref{E:V_good} we see that, as processes,
\begin{equation}\label{E:wealth_repre}
x+ V(Q) = \frac{1}{\gamma}\log(X(\pi)).
\end{equation}
In other words, there is a direct connection between the wealth process (induced by $\pi$) in the fictitious market \eqref{E:S_dynamics} and the gains process (induced by $Q$) in the market with price impact. The requirement $\pi_t = H_t(Q_t)-H_t(0)$  is similar to the invertibility conditions in  \cite[Condition 3.7]{MR3778361}, \cite[Section 4]{fukasawa18}, and the bounded-ness condition in \cite[Proposition 3.5]{sah2014hedging}.

Clearly, for $Q\in\A_{PI}$ we can construct $\pi$. To go in the reverse direction, we must have for $\leb_{[0,T]}$ almost every $t\leq T$ that $\prob$ almost surely $\pi_t$ lies in the random constraint set $\K^o_t$, where
\begin{equation}\label{E:constraint_set}
\begin{split}
\K_t &\dfn \cbra{H_t(q)\such q\in\reals^k};\qquad \K^o_t \dfn \cbra{H_t(q)-H_t(0)\such q\in \reals^k}.
\end{split}
\end{equation}
Therefore, we define the acceptable strategies
\begin{equation*}
\begin{split}
\mathcal{A}_{C}\dfn \cbra{\pi \in \mcp(\filt) \such \int_0^T |\pi_t|^2 dt < \infty, \pi \in \K^o, \ \leb_{[0,T]}\times\prob \textrm{ a.s.}}.
\end{split}
\end{equation*}
In \cite[Lemma 5.5]{MR3005018} the map $q\to H(q)$ was shown to be continuous (in fact, the map is much more regular: see \cite{PulidoKramkov17}). Thus, by the Kuratowski-Ryll-Nardzewski measurable selection theorem for any $\theta\in\mathcal{P}(\filt)$ such that $\theta_t \in \K_t$ a.s. for $t\leq T$, we can select $Q\in\mathcal{P}(\filt)$ so that $H(Q) = \theta$, $\leb_{[0,T]}\times\prob$ a.s.. As such we have the equivalence $Q\in\A_{PI} \Leftrightarrow \pi\in \A_{C}$ for $\pi = H(Q)-H(0)$, and that $x+V(Q) = (1/\gamma)\log\left(X(\pi)\right)$.

\subsection{The constraint set} \label{SS:cons_set} The constraint-set process $\K^o$ of \eqref{E:constraint_set} is determined by $\gamma,\Sigma_0$ and $\Psi$. In this section we construct three examples which show generally, for fixed $(t,\omega)$, that $\K^o_t(\omega)$ is neither closed or convex.  This departs from the literature of utility maximization problem under constraints, where both closed-ness and convexity are assumed, and creates a severe problem when it comes to solving the investor's optimal investment problem under price impact, as will be discussed in the next section. In the first example there is no constraint as $\K^o_t = \reals^d$; in the second $\K^o_t$ is not closed; and in the third $\K^o_t$ is not convex.

\begin{exa}\label{Ex:cons_bach} \textit{Bachelier model: no constraint}.  Similarly to \cite[Example 4.11]{MR3375887}, let $k=d$, $\Sigma_0 = \int_0^T f_t'dB_t$ and $\Psi = \int_0^T \psi_t' dB_t$, where $f\in L^2([0,T];\reals^d)$ and $\psi\in L^2([0,T];\reals^{d\times d})$.  Assume $\psi_t$ is invertible, and the map $t\to \psi_t^{-1}$ is continuous. Clearly, Assumption \ref{A: integrability} holds, and calculation shows
\begin{equation*}
\frac{N_t(q)}{N_0(q)} = \mathcal{E}\left(-\gamma\int_0^t (f_s+\psi_sq)'dB_s\right);\qquad t\leq T.
\end{equation*}
Thus $H_t(q) = -\gamma(f_t + \psi_tq)$, and hence $\K_t = \K^o_t = \reals^d$ with $\pi_t = H_t(Q_t)-H_t(0)$ if and only if $Q_t = -(\gamma\psi_t)^{-1}\pi_t$.
\end{exa}

\begin{exa}\label{Ex:cond_dig} \textit{Digital claim: a non-closed set}. Consider when $k=d=1$, $\Sigma_0 = 0$, and $\Psi = 1_{B_T\geq 0}$.  Here we have on $\cbra{B_t=b}$ and for $\tau=T-t$:
\begin{equation}\label{E:cons_set_digital}
\K_t = \K^o_t = \frac{1}{\sqrt{\tau}}\varphi\left(\frac{b}{\sqrt{\tau}}\right)\times \left(-\frac{1}{\Phi\left(\frac{b}{\sqrt{\tau}}\right)}, \frac{1}{1-\Phi\left(\frac{b}{\sqrt{\tau}}\right)}\right),
\end{equation}
where $\Phi$ is the cumulative distribution function of a standard normal random variable, and $\varphi$ its probability density function.  The left endpoint above occurs as $q\downarrow -\infty$, while the right endpoint occurs as $q\uparrow\infty$. To verify \eqref{E:cons_set_digital}, we note for $\Sigma_0 =0$, $\Psi = f(B_T)$ with $f$ bounded and measurable, it follows that $H_t(q) = -\partial_b v(t,B_t;q)$, where $v(t,b;q)\dfn -\log\left(\espalt{}{}{e^{-\gamma qf(B_T)}\such B_t = b}\right)$.  Indeed, this can be shown using It\^{o}'s formula and the transition density for $B$.  In this instance
\begin{equation*}
v(t,b;q) = -\log\left(e^{-\gamma q} + \Phi\left(-\frac{b}{\sqrt{T-t}}\right)\left(1-e^{-\gamma q}\right)\right).
\end{equation*}
The result follows by differentiation, the monotonicity of $-\partial_b v(t,B_t;q)$ with respect to $q$, and taking $q\uparrow\infty$, $q\downarrow-\infty$.
\end{exa}

\begin{exa}\label{Ex:cons_dig_2d} \textit{Two dimensional digital and linear claim: a closed but non-convex set.} Consider when $k=d=2$, $\Sigma_0 = 0$ and $\Psi = \left(B^{1}_T, B^{1}_T 1_{B^{2}_T\geq 0}\right)$.  Here, with $\tau=T-t$, and on $\cbra{B^{1}_t=b_1,B^{2}_t = b_2}$ define $A = \Phi\left(-b_2/\sqrt{\tau}\right)$ and $C = \sqrt{\tau}/\phi\left(-b_2/\sqrt{\tau}\right)$.
Then, $\K_t = \K^o_t = \cbra{(p_1,p_2)}$ where $p_1$, $p_2$ must satisfy
\begin{equation*}
\begin{split}
-\frac{1}{(1-A)C} &< p_2 < \frac{1}{AC};\quad \left(\tau p_1-b_1\right)^2 \geq 2\tau\left(1-2(1-ACp_2)(1-A)\right)\log\left(\frac{1+(1-A)Cp_2}{1-ACp_2}\right).
\end{split}
\end{equation*}
It is easy to see $\K^o_t$ is not convex. Indeed, choose $p_2$ close enough to $1/(AC)$ to make the right hand side of the second equality above strictly positive.  Then, $(p_1,p_2)\not\in\K^o_t$ for $p_1$ near $b_1/\tau$ while $(p_1,p_2)\in\K^o_t$ for $|p_1|$ large enough. This is shown in Figure \ref{F:2d_cons}. Therein, when solving $H(q) = p$, in the interior of the shaded region, there are two solutions; along the boundary, but not where the lines cross, there is a unique solution; and where the lines cross there is an uncountable family of solutions.  This shows the constraint set is closed, but also that typically there will not be a unique optimal demand policy for the large investor's optimal investment problem.
\begin{figure}
\includegraphics[height=5cm,width=8cm]{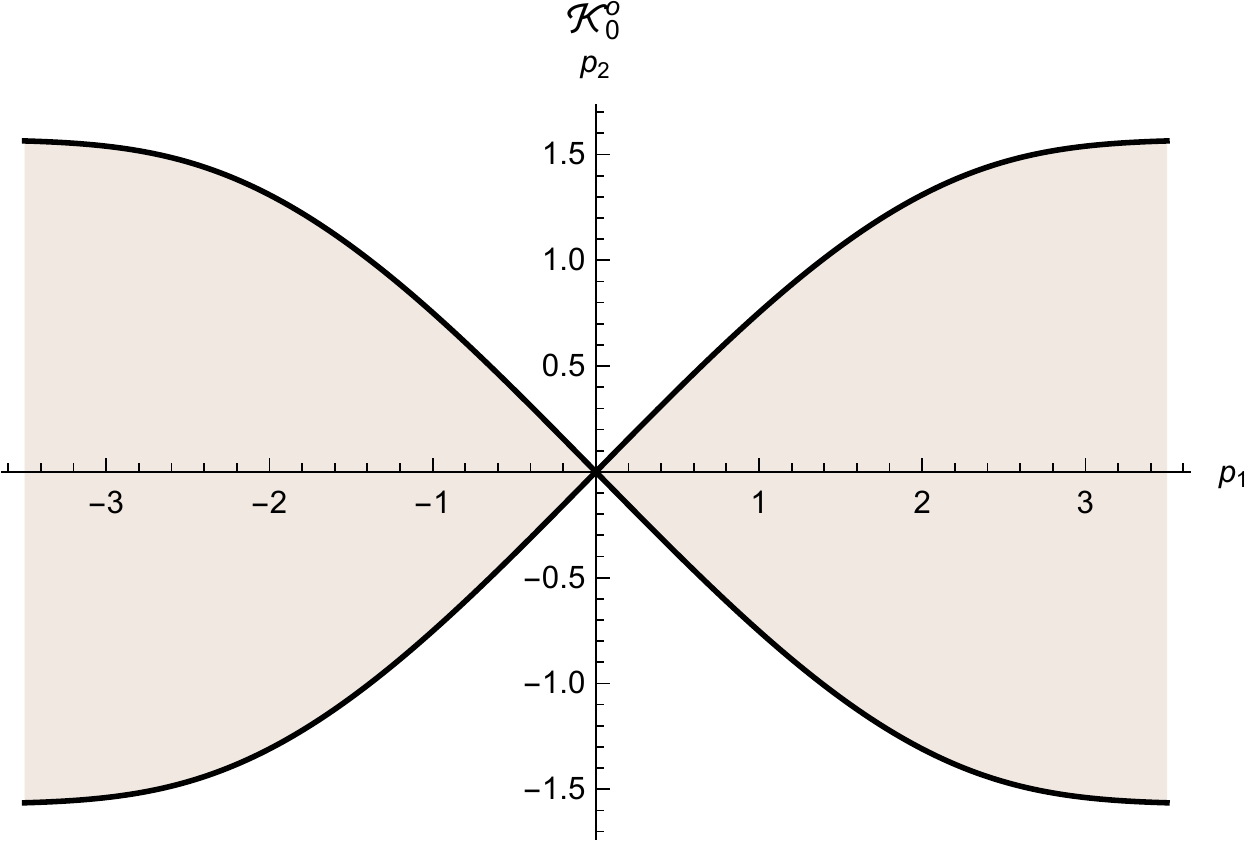}
\caption{$\K^o_0$ for the two-dimensional claim of Example \ref{Ex:cons_dig_2d}.} \label{F:2d_cons}
\end{figure}
\end{exa}

\section{Optimal Investment with Price Impact}\label{S:cara_li}


We now consider the investor's optimal investment problem. We assume she has preferences described by a utility function $U$ defined on the real line, which is $C^2$, strictly concave, and satisfies the conditions of reasonable asymptotic elasticity (c.f.~\cite{MR1865021}). In view of Lemma \ref{L:PI_gains_process}, for $Q\in\A_{PI}$ the wealth process under price impact is well defined, and in turn we can define her value function, for a given initial capital $x$ and endowment $\Sigma_1$ by
\begin{equation}\label{E:PI_opt}
u(x;\Sigma_1) \dfn \sup_{Q\in\mathcal{A}_{PI}}\espalt{}{}{U\left(x+ V_T(Q) + \Sigma_1 \right)}.
\end{equation}
To ensure $u(x;\Sigma_1)$ is well defined, we assume
\begin{ass}\label{A: large_investor_integ_abs}
$\espalt{}{}{\left(U(x+\Sigma_1)\right)^{-}} < \infty$ for all $x\in\reals$.
\end{ass}
Based upon the results of Section \ref{SS:constrain}, we can equate $u(x;\Sigma_1)$ with the value function for an optimal investment problem in the fictitious model with constrained trading, where the investor's preferences are described by the utility field
\begin{equation*}
\tilde{U}(w,\omega) \dfn U\left(\frac{1}{\gamma}\log(w) + \Sigma_1(\omega)\right);\qquad w>0, \omega\in\Omega.
\end{equation*}
It is easy to verify that for $\omega$ fixed, $\tilde{U}(w,\omega)$ is strictly increasing, concave, and satisfies the asymptotic elasticity conditions in \cite{MR1722287}. Furthermore we know that $Q\in \A_{PI} \Leftrightarrow \pi\in\A_{C}$ for $\pi = H(Q)-H(0)$, and using \eqref{E:wealth_repre} we conclude, for initial wealth $x$ in the price impact model and $e^{\gamma x}$ in the fictitious model
\begin{equation*}
\tilde{U}\left(X_T(\pi),\cdot\right) = U\left(\frac{1}{\gamma}\log\left(X_T(\pi)\right) + \Sigma_1\right) = U\left(x + V_T(Q) + \Sigma_1\right).
\end{equation*}
For comparison purposes, we also define the unconstrained problem in the fictitious market, setting
\begin{equation*}
\begin{split}
\mathcal{A}&\dfn \cbra{\pi \in\mathcal{P}(\filt)\such  \int_0^T |\pi_t|^2dt < \infty \textrm{ a.s.}}.
\end{split}
\end{equation*}
Using $\tilde{U}$ in the fictitious market, we define the value functions
\begin{equation}\label{E:opt_cons_no_cons}
\begin{split}
\tilde{u}(x;\Sigma_1) &\dfn \sup_{\pi\in\mathcal{A}} \espalt{}{}{\tilde{U}\left(X_T(\pi),\Sigma_1\right)\such X_0 = e^{\gamma x}};\\
\tilde{u}_C(x;\Sigma_1) &\dfn \sup_{\pi\in\mathcal{A}_C} \espalt{}{}{\tilde{U}\left(X_T(\pi),\Sigma_1\right)\such X_0 = e^{\gamma x}},
\end{split}
\end{equation}
and based on the above, we immediately have
\begin{prop}\label{T:PI_to_Cons}
Let Assumptions \ref{A: integrability} and \ref{A: large_investor_integ_abs} hold, and for $x\in\reals$ let $u(x;\Sigma_1)$ be from \eqref{E:PI_opt}. Then
$$u(x;\Sigma_1) = \tilde{u}_C(x;\Sigma_1) \leq \tilde{u}(x;\Sigma_1).$$
\end{prop}

The unconstrained utility maximization problem for random endowments and utility fields has been well studied. For endowments with utility functions on $(0,\infty)$ we highlight \cite{MR2052905},  while for utility fields, we pay particular attention to \cite{MR3292127}, which contains necessary and sufficient conditions for existence of optimal policies, in the general incomplete setting. Presently, the analysis greatly simplifies as the unconstrained market is complete.  Indeed,
\begin{equation}\label{E:No_cons_opt_static}
\tilde{u}(x;\Sigma_1) = \sup_{\xi \in L^0(\F_T)}\cbra{\espalt{}{}{U(x+\xi+\Sigma_1)} \such \espalt{}{0}{e^{\gamma \xi}}\leq 1},
\end{equation}
where $\mathbb{E}_0$ is expectation with respect to $\qprob_0$ defined by (recall \eqref{E:rn_deriv_lambda} and $\lambda = -H(0)$)
\begin{equation}\label{E:rn_deriv}
\left.\frac{d\qprob_0}{d\prob}\right\vert_{\F_T} = \frac{e^{-\gamma \Sigma_0}}{\espalt{}{}{e^{-\gamma\Sigma_0}}} = \mathcal{E}\left(\int_0^\cdot H_t(0)'dB_t\right)_T.
\end{equation}
To see this, for $\pi\in \A$ set $\xi = (1/\gamma)\log(X_T(\pi)) - x$.  Alternatively, for $\xi$ such that $\espalt{}{0}{e^{\gamma\xi}}\leq 1$ there is $\pi\in\mcp(\filt)$ satisfying $\mathcal{E}\left(\int_0^\cdot \pi_t'(dB_t-H_t(0)dt)\right)_T = e^{\gamma \xi}/\espalt{}{0}{e^{\gamma\xi}}$, and $(1/\gamma)\log(X_T(\pi)) \geq x+\xi$.  We record for later use that \eqref{E:V_good} implies
\begin{equation}\label{E:wealth_bc}
\espalt{}{0}{e^{\gamma V_T(Q)}}\leq 1,
\end{equation}
and this budget constraint has major ramifications. Now, the utility maximization problem with random constraints is also by now standard.  However, most of the literature (c.f.~\cite{cvitanic1992cdc}, \cite{MR1640352}), assumes the constraint set is both closed and convex a.s.~for $t\leq T$.  Closed-ness is used to obtain an optimal policy when the constraint binds, and convexity is used when constructing the auxiliary markets with no constraints.  However, as shown in Section \ref{SS:cons_set}, in general one cannot assert $\K^o_t$ of \eqref{E:constraint_set} has either of these properties.

\begin{rem}\label{R:dichotomy} In light of the above, there is a dichotomy when solving the optimal investment problem in this class of price impact models.  Indeed, under very mild conditions upon $U$ and $\Sigma_1$ (see, for example, \cite{MR2489605}), there is an optimal policy $\hat{\pi}$ to the unconstrained problem in \eqref{E:opt_cons_no_cons}.  Then, in the constrained case, we have two alternatives:
\begin{itemize}
\item [\textit{I.}] $\hat{\pi}\in\K^o$, $\Leb_{[0,T]}\times \prob$ almost surely.
\item [\textit{II.}] There exists a set $(a,b)\times E\in [0,T]\times \F_T$ with $a < b$ and $\prob\bra{E}>0$ s.t.~$\hat{\pi}_t \not\in \K^o_t$ on $E\times (a,b)$.
\end{itemize}
In case \textit{I.}, the constrained problem admits the same answer as the unconstrained problem, and the constraint is effectively absent. In case \textit{II.}, the optimal investment problem falls outside the scope of most of the existing literature. In addition, in \textit{II.}, one may end up with a strange situation where the investor is induced to demand an infinite number of shares from the market maker, and the market maker, at least formally, is willing to quote a finite cash balance for this demand. For example, such a situation occurs in Example \ref{Ex:cond_dig}, when $\Sigma_1$ is formed precisely so that $\hat{\pi}$ is the left endpoint in \eqref{E:cons_set_digital} for $t\in (a,b)$.
\end{rem}

\subsection{Exponential preferences} The usefulness of Proposition \ref{T:PI_to_Cons} is highlighted when the investor has exponential utility. Here, exponential utility in the price impact model is connected to power utility in the fictitious constrained market. To proceed, we denote by $\alpha>0$ the absolute risk aversion of the investor and specify Assumption \ref{A: large_investor_integ_abs} to
\begin{ass}\label{A: large_investor_integ}
$\espalt{}{}{e^{-\alpha\Sigma_1}} < \infty$.
\end{ass}

Assumption \ref{A: large_investor_integ} allows us to employ the usual change of measure (c.f.~\cite{MR1891730}) with the endowment $\Sigma_1$, defining $\tprob$ on $\F_T$ by
\begin{equation}\label{E:tprob_def}
\left.\frac{d\tprob}{d\prob}\right\vert_{\F_T} = \frac{e^{-\alpha\Sigma_1}}{\espalt{}{}{e^{-\alpha\Sigma_1}}}.
\end{equation}
Using $\tprob$ and that $e^{\alpha((1/\gamma) \log(X_T(\pi)))} = X_T(\pi)^{-\alpha/\gamma}$, straight-forward calculations imply the following result (where $\wt{\mathbb{E}}$ stands for the expectation under measure $\tprob$).
\begin{prop}\label{P:power} Let Assumptions \ref{A: integrability} and \ref{A: large_investor_integ} hold and recall $\tilde{u}$ and $\tilde{u}_C$  from \eqref{E:opt_cons_no_cons}. Then, taking $U(x)=-e^{-\alpha x}$, we have
\begin{equation}\label{E:cons_no_cons_opt_power}
\begin{split}
\tilde{u}(0;\Sigma_1) &=  \frac{\alpha}{\gamma}\espalt{}{}{e^{-\alpha\Sigma_1}}\left(\sup_{\pi\in\mathcal{A}} \tespalt{}{}{\frac{1}{p}\left(X_T(\pi)\right)^p\such X_0 = 1}\right);\\
\tilde{u}_C(0;\Sigma_1) &=  \frac{\alpha}{\gamma}\espalt{}{}{e^{-\alpha\Sigma_1}}\left(\sup_{\pi\in\mathcal{A}_C} \tespalt{}{}{\frac{1}{p}\left(X_T(\pi)\right)^p\such X_0 = 1}\right),\\
\end{split}
\end{equation}
where $p\dfn-\alpha/\gamma$.
\end{prop}



\subsection{Solving the optimal investment problem}

As we have just seen, price impact manifests itself in two ways. The first has to do with market incompleteness, induced when $K^o_t \neq \reals^d$. The second way is due to the altered budget constraint from \eqref{E:wealth_bc}, which differs from the ``traditional'' constraint   $\espalt{\qprob_0}{}{V(Q)_T}\leq 0$ in frictionless markets.


Therefore, even when the constraint is absent (i.e.~$K^o_t = \reals^d$ almost surely on $t\leq T$) and the market is complete, there is a non-trivial ``budget constraint'' effect due to price impact. Our interest lies in this budget constraint impact. Therefore, for the remainder of the article, we assume the constraint is non-binding (effectively absent). Keeping Assumptions \ref{A: integrability}, \ref{A: large_investor_integ} in force, H\"{o}lder's inequality ensures
\begin{equation}\label{E:opt_m_integ}
\espalt{}{}{e^{-\beta\left(\Sigma_1+\Sigma_0\right)}} < \infty;\qquad \frac{1}{\beta} \dfn \frac{1}{\alpha} + \frac{1}{\gamma},
\end{equation}
so $\beta$ is the representative risk aversion of the investor and market maker. Therefore, there exists $\hphi\in\mcp(\filt)$ verifying
\begin{equation}\label{E:opt_M_uncon}
\frac{e^{-\beta\left(\Sigma_1 + \Sigma_0\right)}}{\espalt{}{}{e^{-\beta\left(\Sigma_1+\Sigma_0\right)}}}=\mathcal{E}\left(\int_0^\cdot  \hphi_t'dB_t\right)_T.
\end{equation}
We make the following assumption on $\Sigma_0,\Psi$ and $\Sigma_1$, which ensures the constraint $\hat{\pi}\in\K^0$ is non-binding:
\begin{ass}\label{A:H2}
$\hphi\in\K$, $\Leb_{[0,T]}\times\prob$ almost surely.
\end{ass}

Assumption \ref{A:H2} states an implicit connection between endowments $\Sigma_0$ and $\Sigma_1$ and the tradeable assets $\Psi$. As shown below, it holds in many cases of practical interest, and clearly it causes the constrained and unconstrained optimization problems to coincide. Therefore we have the following result, the proof of which is in Appendix \ref{AS:Proofs}.

\begin{prop}\label{P:powerUnconstrained}
Let Assumptions \ref{A: integrability}, \ref{A: large_investor_integ} and \ref{A:H2} hold. Then, for each of the optimal investment problems in \eqref{E:cons_no_cons_opt_power}, the optimal trading strategy $\hat{\pi}\in\A_C$ is $\hat{\pi}  = \hphi - H(0)$, and the value functions \eqref{E:PI_opt} and \eqref{E:cons_no_cons_opt_power} coincide, taking the explicit value
\begin{equation}\label{E:vf_uncon}
\begin{split}
- \espalt{}{}{e^{-\gamma \Sigma_0}}^{-\frac{\alpha}{\gamma}} \times \espalt{}{}{e^{-\beta\left(\Sigma_0 + \Sigma_1\right)}}^{\frac{\alpha}{\beta}}.
\end{split}
\end{equation}
\end{prop}

\subsection{Indicative examples}\label{S:Examples}

There are concrete examples where Assumption \ref{A:H2} holds. One easy example is when $\K^o=\reals^d$, $\Leb_{[0,T]}\times\prob$ almost surely: as in the Bachelier model of Example \ref{Ex:cons_bach}. We now provide additional examples.

\begin{exa}\label{Ex:MarkovModel} \textit{Markov models.}  Assume $\Sigma_0=\Sigma_{0}(Y_{T})$ and $\Psi = \Psi(Y_{T})$ where $Y$ is a diffusion with dynamics $dY_{t}=b(t,Y_{t})dt+\sigma(t,Y_{t})dB_{t}$ and $Y_{0}=y_{0}$. Here $Y$ is understood as an economic factor. Letting $F(y;q)\dfn -\gamma\Sigma_0(y) - \gamma q'\Psi(y)$, we notice, using the Markov property, that $N_t(q)=v(t,Y_{t};q)$ where $v(t,y;q)\dfn\mathbb{E}\left[e^{F(Y_{T};q)}|Y_t=y\right]$. For $q$ fixed, $v(\cdot;q)$ solves the Cauchy problem
\begin{align*}
v_{t}+\mathcal{L}v&=0, \text{ on }[0,T)\times\mathbb{R}^{d};\qquad v(T,\cdot) =e^{F(\cdot;q)}, \text{ on }\mathbb{R}^{d},
\end{align*}
where $\mathcal{L}$ is the infinitesimal generator of $Y$. Assume $b,\sigma, \Sigma_0,\Psi$ are regular enough so there exists a $u\in C^{1,2,0}$ that solves the above Cauchy problem (see for example \cite{MR0181836} for general regularity results). Using It\^{o}'s formula
\begin{align*}
dN_{t}(q)&=N_{t}(q)\nabla_{y}\log v(t,Y_{t};q) \sigma(t,Y_{t})dB_{t}.\nonumber
\end{align*}
Thus, $H_{t}(q)=\sigma(t,Y_{t})'\nabla_{y}\log v(t,Y_{t};q)$, and the validity of Assumption \ref{A:H2} depends on smoothness and monotonicity properties of the map  $q\mapsto \nabla_{y}\log v(t,Y_{t};q)$. As a matter of fact, Theorem 3 of \cite{MR3778361} shows when $d=k=1$;  $b(t,y)=b(t)$; $\sigma(t,y)=\sigma(t)$; $\Sigma_{0}$ and $\Psi$ are of linear growth; and $\Psi$ strictly monotone in $\mathbb{R}$, then $\K_t=\K^o_t=\mathbb{R}$. Interestingly, \cite{MR3778361} also  provides an example where $\Psi$ is the payoff of a European call (which is not strictly monotone) and $\K_{t}\neq\mathbb{R}$. In fact, in this Markovian setting, identifying $\K_t$ is equivalent to identifying, for $t\leq T, y\in\reals^d$ fixed, the range $\cbra{\sigma(t,y)'\nabla_{y}\log v(t,y;q)\such q\in\reals^k}$.  To the best of our knowledge, general results of this type are absent in the literature.
\end{exa}

A common assumption in the literature is that $\Psi$ is linear in $\Sigma_0,\Sigma_1$ (see \cite{AdmPfl88, GloMil85, Kyl89}, as well as the more recent \cite{KylObiWa18, RosWer15}). In other words, agents securitize their risky positions, and through trading achieve a mutually beneficial risk reduction.  The following examples are related to this setting, and verify Assumption \ref{A:H2}. Furthermore, within this setting, dynamic trading is not necessary since the optimal strategy is static.

\begin{exa}\label{Ex:IndependentCase} \textit{Endowments as portfolios of $\Psi$ and an independent component.} Let $d=k+n$ and write $\ol{B}=(B^1,...,B^k)$ and $\ul{B} = (B^{k+1},...,B^{k+n})$. For $k_0,k_1\in\reals^k$, assume $\Sigma_{0}=k_{0}'\Psi+Y_{0}$, $\Sigma_{1}=k_{1}'\Psi+Y_{1}$, where $\Psi$ is $\F^{\ol{B}}_T$ measurable and $Y_0,Y_1$ are $\F^{\ul{B}}_T$ measurable. $Y_0$ and $Y_1$ can be interpreted as the idiosyncratic components of the respective endowments; while the positions in $\Psi$ are the part invested in the tradeable asset (for the market maker, the latter can be thought as the aggregate inventory).  As $(Y_{0}, Y_{1})$ and $\Psi$ are independent given $\F_t$ for all $t\leq T$, it follows that for $q\in\reals^k$ and $t\leq T$
\begin{equation*}
N_t(q) = \condespalt{}{e^{-\gamma(k_{0}+q)'\Psi}}{\F_t}\condespalt{}{e^{-\gamma Y_0 }}{\F_t}.
\end{equation*}
In \eqref{E:Nq}, write $H(q)=\left(\ol{H}(q), \ul{H}(q)\right)$. We see
\begin{equation*}
\E\left(\int_0^\cdot \ol{H}_u(q)'d\ol{B}_u\right)_T=\frac{e^{-\gamma(k_{0}+q)'\Psi}}{\espalt{}{}{e^{-\gamma (k_{0}+q)'\Psi}}};\qquad \E\left(\int_0^\cdot \ul{H}_u(q)'d\ul{B}_u\right)_T=\frac{e^{-\gamma Y_{0}}}{\espalt{}{}{e^{-\gamma Y_{0}}}}.
\end{equation*}
As $\ul{H}$ does not depend upon $q$, we find $\mathcal{K}^o_t= \left( \cbra{\ol{H}_t(q)}_{q\in\reals},0\right)$. Next, recall $\hphi=\left(\ol{\hphi},\ul{\hphi}\right)$ from \eqref{E:opt_m_integ} and \eqref{E:opt_M_uncon}.  Due to the independence
\begin{equation*}
\mathcal{E}\left(\int_0^\cdot \hphi_t'dB_t\right)_T =
 \frac{e^{-\beta\left(k_{0} + k_{1}\right)'\Psi}}{\espalt{}{}{e^{-\beta\left(k_{0} + k_{1}\right)'\Psi}}}\frac{e^{-\beta\left(Y_{0} + Y_{1}\right)}}{\espalt{}{}{e^{-\beta\left(Y_{0} + Y_{1}\right)}}} = \mathcal{E}\left(\int_0^\cdot \ol{\hphi}_t'd\ol{B}_t\right)_T\mathcal{E}\left(\int_0^\cdot \ul{\hphi}_t'd\ul{B}_t\right)_T.
\end{equation*}
Since
\begin{equation*}
-\beta(\Sigma_0 + \Sigma_1) = -\gamma \Sigma_0 -\gamma\frac{\alpha \Sigma_1 - \gamma \Sigma_0}{\alpha+\gamma},
\end{equation*}
we see that $\ol{\hphi} = \ol{H}((\alpha k_1 - \gamma k_0)/(\alpha+\gamma))$, but we can only have $\ul{\hphi} = \ul{H}(0)$ (i.e.~Assumption \ref{A:H2}) when $\alpha Y_1 = \gamma Y_0$. Absent this condition, it fails. However, the constraint $\pi \in \K^o$ means we cannot invest in $\ul{S}$, while the optimal demand in $\ol{S}$ is $\ol{\hphi}-\ol{H}(0)$. Thus, for the investor, the static position $\hat{Q}_t = (\alpha k_1-\gamma k_0)/(\alpha+\gamma)$, $t\leq T$ is optimal, and there is no need for dynamic trading. Indeed, the optimal order puts the market maker and investor at a Pareto-optimal situation (see \cite{MR2211128}), and there is no other mutually beneficial transaction to be made.
\end{exa}

\begin{exa}\label{Ex:BachelierModel} \textit{Bachelier model revisited.} Our last example shows dynamic optimal trading strategies are possible. We extend the Bachelier model of  Example \ref{Ex:cons_bach} by setting $\Sigma_1 = \int_0^T g_t'dB_t$ where $g\in L^2([0,T];\reals^d)$. Assumption \ref{A:H2} holds as $\K_t = \reals^d$, and calculation shows $\hphi$ and $\hat{\pi}$ from \eqref{E:opt_M_uncon}, Proposition \ref{P:powerUnconstrained} respectively are
\begin{equation*}
\hphi = -\beta\left(f + g\right);\quad \hat{\pi} = \hphi - H(0) = \frac{\gamma}{\alpha+\gamma}\left(\gamma f - \alpha g\right).
\end{equation*}
Therefore, $\hat{Q} =  (1/(\alpha+\gamma))\psi^{-1}\left(\alpha g - \gamma f\right)$. Note that the optimal demand is deterministic, but not necessarily static as it was in Example \ref{Ex:IndependentCase}. However, when $f = \psi k_0$ and $g = \psi k_1$ (endowments are portfolios of $\Psi$) we recover again static optimal positions.

\end{exa}

\section{Contingent Claim Analysis}\label{S:Pricing}

We arrive at our central topic: obtaining pricing rules and demand schedules for contingent claims. Our focus is on the budget constraint aspect of price impact, and hence we assume $\K^0_t =\reals^d$. We first ask what it means for a price to be arbitrage-free, since the budget constraint in \eqref{E:No_cons_opt_static} is clearly not linear in the payoff $\xi$. This implies non-scalability of arbitrages and hence the topic of the arbitrage-free contingent claim pricing must be revisited.

We next address the problem of optimal demand by formulating the demand schedule, which gives how many units of the claim the investor wants to hold at each price (arbitrage-free or not). As we will see, it is exactly because of the price-impact that even traded prices which create arbitrage opportunities  do not yield infinite demands (i.e.~arbitrage is limited).

We then endogenize the traded price as an equilibrium one between two investors who bilaterally trade the claim, but hedge in segmented markets. Interestingly, there are cases where the equilibrium traded price allows for arbitrage.

The claim has terminal payoff $h\in\reals$. In contrast to the securities with terminal payoff $\Psi$, $h$ is not traded through the market maker: the investor trades $h$ with another investor, and then hedges her position by dynamically trading with the market maker. For example, $h$ can be seen as the payoff of an illiquid derivative, which is not traded through the market
makers but rather it is traded among two investors or financial institutions (illiquity on some derivative markets has also been supported by empirical studies, e.g.~\cite{BreEldHau01, ChrGoyJacKar18, DeuGupSub11}). To isolate the budget constraint effect, and satisfy all integrability requirements, we assume
\begin{ass}\label{A:H4}
$\K^o_t = \reals^d$ almost surely for $t\leq T$. Furthermore,
$\espalt{}{}{e^{-\gamma\Sigma_0 + p(|\Psi|+|h|)}} <\infty$ and $\espalt{}{}{e^{-\alpha\Sigma_1 + p|h|}} <\infty$ for all $p\geq 0$.
\end{ass}


\subsection{Arbitrage-free prices}\label{SS:arb_free_px}

Fix a position size $\qtau\in\reals$. Assumption \ref{A:H4} implies the investor can hedge the claim $\qtau h$. Indeed, with $\qprob_0$ from \eqref{E:rn_deriv}, since $\espalt{}{0}{e^{\gamma\qtau h}} < \infty$, there is a demand process $Q\in\A_{PI}$ and a (per unit) initial capital $\ol{h}(\qtau)$ such that $\qtau \ol{h}(\qtau)+ V_T(Q) = \qtau h$ a.s.. This follows by predictable representation, which asserts the existence of a strategy $\pi\in\mcp(\filt)$ so that
\begin{equation}\label{E:short_rep}
\frac{e^{\gamma\qtau h}}{\espalt{}{0}{e^{\gamma\qtau h}}} = \mathcal{E}\left(\int_0^\cdot \pi_t'(dB_t-H_t(0)dt)\right)_T = e^{\gamma V_T(Q)},
\end{equation}
where the last equality follows solving $\pi_t=H_t(Q_t)-H_t(0)$ for $Q_t$. The required initial capital is
\begin{equation}\label{E:short_rep_cap}
\ol{h}(\qtau)\dfn \frac{1}{\gamma\qtau}\log\left(\espalt{}{0}{e^{\gamma\qtau h}}\right),
\end{equation}
which is exactly the market maker's indifference (per unit) value for selling $\qtau$ units of $h$. As $h$ is not traded with the market maker, we want to identify prices for $h$ which preclude arbitrage opportunities from arising when the investor trades $\Psi$ with the market maker. Despite the completeness implied by $\K^o_t = \reals^d$, due to price impact, there need not be only one ``arbitrage-free'' price for $\qtau h$. In fact, we presently develop three notions of an arbitrage-free price, which coincide in the frictionless case, but give strikingly different answers considering price impact.

To develop the first two notions, we see from \eqref{E:short_rep_cap} that in contrast to the frictionless case, the replicating capital for $-\qtau h$ is not $-\qtau \ol{h}(\qtau)$. Rather, it is $-\qtau\ul{h}(\qtau)$ where
\begin{equation}\label{E:long_rep_cap}
\ul{h}(\qtau)\dfn -\frac{1}{\gamma\qtau } \log\left(\espalt{}{0}{e^{-\gamma\qtau h}}\right),
\end{equation}
is the market maker's indifference (per unit) value for buying $\qtau$ units of claim $h$. Jensen's inequality implies $\ul{h}(\qtau)\leq \ol{h}(\qtau)$, which simply verifies the fact that buying value (bid) is less than the corresponding selling value (ask).

We first define a strong notion of an arbitrage-free price for $h$, consistent with that in \cite{MR3512788}:
\begin{defn}\label{D:strong_arb_free_price}
$p$ is an \emph{arbitrage-free for all positions} in $h$, provided for all $Q\in\A_{PI}$ and $\qtau\in\reals$, if $\qtau p+V_T(Q)  - \qtau h\geq 0$ a.s., then $\qtau p+V_T(Q) -\qtau h = 0$ a.s..
\end{defn}
The reasoning behind Definition \ref{D:strong_arb_free_price} is clear: we rule out the investor being able to start with nothing, and for some $\qtau\in\reals$, trade $\qtau$ units of $h$ for a per-unit price of $p$; trade $\Psi$ with the market maker over $[0,T]$; and cover her position with positive probability of gain at $T$.

The above definition is strong as it rules out arbitrages for \textit{all} position sizes $\qtau$ simultaneously. A weaker notion of an arbitrage-free price is:
\begin{defn}\label{D:arb_free_price}
$p$ is an \emph{arbitrage-free at the level} $\qtau >0$ in $h$ provided
\begin{enumerate}
\item [\textit{(a)}] For all $Q\in\A_{PI}$, if $\qtau p+V_T(Q)  - \qtau h\geq 0$ a.s., then $\qtau p+V_T(Q) -\qtau h = 0$ a.s..
\item [\textit{(b)}] For all $Q\in \A_{PI}$, if $-\qtau p+V_T(Q) + \qtau h\geq 0$ a.s., then $-\qtau p+V_T(Q)+ \qtau h = 0$ a.s..
\end{enumerate}
\end{defn}

In other words, an arbitrage-free price $p$ for at the level $\qtau>0$ rules out arbitrages for either buying or selling $\qtau$. Based on these two definitions we obtain the following, proved in Appendix \ref{AS:Proofs}.
\begin{prop}\label{P:arb_free_range}
Let Assumption \ref{A:H4} hold. Then,
\begin{itemize}
\item [\textit{i)}] The range of arbitrage-free prices for $h$ in the sense of Definition \ref{D:strong_arb_free_price} is the singleton $\espalt{}{0}{h}$.
\item [\textit{ii)}] For any fixed $\qtau > 0$, the range of arbitrage-free prices for $h$ in the sense of Definition \ref{D:arb_free_price} is the closed interval  $[\ul{h}(\qtau),\ol{h}(\qtau)]$ from \eqref{E:short_rep_cap} and \eqref{E:long_rep_cap} respectively.
\end{itemize}
\end{prop}

As a consequence of Proposition \ref{P:arb_free_range}, there is an arbitrage opportunity whenever the large investor finds an ask (resp.~bid) price lower (resp.~higher) than $\espalt{}{0}{h}$. Indeed, if the asked (resp.~bid) price $p$ is lower (resp.~higher) than $\espalt{}{0}{h}$, there exists a positive number of units $\qtau$ of $h$, for which $p<\ul{h}(\qtau)$ (resp.~$p>\ol{h}(\qtau)$). Then, an arbitrage opportunity arises buying (resp.~selling) $\qtau$ units of $h$, and then hedging by dynamically trading $\Psi$ with the market makers. On the other hand, if the only price for $h$ the investor can find is $\espalt{}{0}{h}$, there is no arbitrage opportunity for any position $\qtau\in\reals$.

Since $\ul{h}(\qtau)$ is decreasing in $\qtau>0$ and $\ol{h}(\qtau)$ is increasing in $\qtau>0$, we immediately see:
\begin{cor}\label{C:arbitrage}
Let Assumption \ref{A:H4} hold. If $p$ is an arbitrage-free price at the level $\qtau>0$, then $p$ is an arbitrage-free price at the level $\qtau'$ for all $ \qtau'\geq \qtau$.
\end{cor}

According to the LaPlace Principle
\begin{equation}\label{E:small_big_h}
\ul{h}(\qtau)\downarrow \ul{h}\dfn \essinf{}{h};\qquad \ol{h}(\qtau)\uparrow\ol{h}\dfn\esssup{}{h}\quad \textrm{as}\quad \qtau\uparrow\infty.
\end{equation}
Therefore, any price $(\ul{h},\ol{h})\ni p\neq \espalt{}{0}{h}$, induces an arbitrage, but the arbitrage opportunity is limited, in the sense there is a maximum number of units of the claim at which the arbitrage can be exploited. In other words, due to price impact, the arbitrage vanishes if the investor takes too large a position. Intuitively, large positions on the claim require large hedging positions in $\Psi$. This changes the market maker's inventory, and the resultant pricing rule is to the detriment of the investor. This sharply departs from arbitrages in markets without price impact, which can arbitrarily scaled.

Definitions \ref{D:strong_arb_free_price} and \ref{D:arb_free_price} do not account for investor preferences. This is an important omission, because the whole discussion on price impact is dedicated to investors large enough to create price impact. We now define a third notion, which states that $p$ is an arbitrage-free price if, given the possibility to purchase arbitrary quantities of $h$ for $p$, the investor's optimal demand is finite. To state it, for a given initial capital $x$, utility function $U$, and endowment $\Sigma_1$, recall the value function $u(x;\Sigma_1)$ from \eqref{E:PI_opt}.
\begin{defn}\label{D:util_arb_free_price}
$p$ is a \textit{utility-demand based arbitrage-free price} for $h$ if for every $\cbra{\qtau_n}_{n\in\mathbb{N}}$ such that
\begin{equation*}
\lim_{n\uparrow\infty} u(x-p\qtau_n, \Sigma_1 + \qtau_n h) = \sup_{\qtau\in\reals} u(-p\qtau;\Sigma_1 + \qtau h),
\end{equation*}
we have $\sup_n |\qtau_n| < \infty$.
\end{defn}

Recall $\beta$ from \eqref{E:opt_m_integ} and $\ol{h},\ul{h}$ from \eqref{E:small_big_h}. For exponential preferences, the following result (with proof in Appendix \ref{AS:Proofs}) shows the range of demand-based arbitrage-free prices is maximal.  This is in stark contrast to the frictionless case (see \cite{MR2212897}).

\begin{prop}\label{P:util_arb_free_price}
Let Assumption \ref{A:H4} hold, and assume the investor has exponential preferences with risk aversion $\alpha$. The range of utility-demand based arbitrage-free prices is $(\ul{h},\ol{h})$.  For $p\in(\ul{h},\ol{h})$, the optimal demand is the unique solution $\hat{\qtau}=\hat{\qtau}(p)$ to
\begin{equation}\label{E:opt_demand}
p = \frac{\espalt{}{}{h e^{-\beta\left(\Sigma_0 + \Sigma_1 + \hat{\qtau} h\right)}}}{\espalt{}{}{e^{-\beta\left(\Sigma_0 + \Sigma_1 + \hat{\qtau} h\right)}}}.
\end{equation}
The map $p\to \hat{\qtau}(p)$ gives the demand schedule.
\end{prop}

Propositions \ref{P:arb_free_range} and \ref{P:util_arb_free_price} imply there are arbitrage prices in the sense of Definitions \ref{D:strong_arb_free_price}, \ref{D:arb_free_price}, at which the investor's optimal demand is finite. Indeed, this holds for all $(\ul{h},\ol{h})\ni p\neq \espalt{}{0}{h}$. However, the arbitrage gain is limited, and the investor may want to trade the claim for hedging purposes. If, for instance, the claim is negatively correlated with her endowment, she has motive to take a long position (even if $p>\espalt{}{0}{h}$, meaning that the arbitrage is provided by a short position). When the benefits from hedging exceed the (limited) arbitrage, the optimal demand does not exploit it (i.e.~the investor optimally ignores the arbitrage opportunity). The following example reinforces this point.

\begin{exa}\label{Ex:BachelierModel2}
We revisit the Bachelier model of Examples \ref{Ex:cons_bach} and \ref{Ex:BachelierModel}, and assume $h=\int_0^T y_t'dB_t$, where $y\in L^2([0,T];\reals^d)$. Assumption \ref{A:H4} readily holds and hence any position $\qtau$ in $h$ can be hedged. Simple calculations show $\espalt{}{0}{h}=-\gamma\int_0^Ty_t'f_tdt$, along with
\begin{equation*}
\ol h(\qtau)=\espalt{}{0}{h}+\frac{1}{2}\gamma\qtau\int_0^T|y_t|^2 dt;\qquad \ul h(\qtau)=\espalt{}{0}{h}-\frac{1}{2}\gamma\qtau\int_0^T|y_t|^2dt.
\end{equation*}
The (linear) optimal demand schedule $p\to \hat{\qtau}(p)$ from \eqref{E:opt_demand} is
\begin{equation}\label{E:optimal demand Bachelier}
\hat{\qtau}(p)= -\frac{p+\beta \int_0^Ty_t'(f_t+g_t)dt}{\beta \int_0^T|y_t|^2dt}= -\frac{p + \alpha\covar{h}{V_T(\hat{Q}(0))+\Sigma_1}}{\beta\var{h}},
\end{equation}
where $\hat{Q}(0)$ is the optimal order flow with (endowment $\Sigma_1$ but) no claim, and the last equality follows using Example \ref{Ex:BachelierModel} at $\Sigma_1=\int_0^T g_t'dB_t$.  Next, a direct calculation specifying Example \ref{Ex:BachelierModel} to $\Sigma_1 = \int_0^T g_t'dB_t + \hat{\qtau}(p)h$ shows  for the optimal strategy $\hat{Q} = \hat{Q}(\hat{\qtau}(p))$
\begin{equation*}
-p\hat{\qtau}(p) + V_T(\hat{Q})  + \hat{\qtau}(p)h = \textrm{Constant  } + \frac{1}{\alpha+\gamma}\int_0^T\left(\gamma f_t + \gamma \hat{\qtau}(p) y_t - \alpha g_t\right)'dB_t.
\end{equation*}
Thus, the large investor does not utilize an arbitrage strategy (save for the pathological case when $\gamma f_t + \gamma\hat{\qtau}(p)y_t -\alpha g_t \equiv 0$ which can only happen if a) $\gamma f_t -\alpha g_t \propto y_t$ and b) the traded price is $p=-\alpha\int_0^T y_t'g_t dt$.)

Relation \eqref{E:optimal demand Bachelier} implies $\hat{\qtau}(p) > 0$ if and only if $p < -\alpha \covar{h}{V_T(\hat{Q}(0)) + \Sigma_1}$. This is intuitive because if $h$ is negatively correlated with the large investor's (optimal) terminal wealth absent the claim, then she has a hedging-related incentive to purchase $h$, and will take long positions at higher prices.  Furthermore, if $\hat{\qtau}(p)> 0$, then, as shown in the proof of Proposition \ref{P:arb_free_range}, in order to rule out arbitrages of the form $-p\hat{\qtau}(p) + V(Q)_T + ph$ we must have $p\geq \ul{h}(\hat{\qtau}(p))$.  Writing $p = -\alpha \covar{h}{V(\hat{Q}(0))_T + \Sigma_1} - \eps$ for some $\eps > 0$ we find (see again Example \ref{Ex:BachelierModel})
\begin{equation*}
p - \ul{h}(\hat{\qtau}(p)) = \frac{\gamma-\alpha}{2\alpha}\eps + \gamma\covar{h}{V_T(\hat{Q}(0))}.
\end{equation*}
For the sake of clarity, assume $\gamma=\alpha$. Here, arbitrages arise (but are not used by the large investor) when $\covar{h}{V_T(\hat{Q}(0))} < 0$. In this instance, the large investor's hedging demands outweigh the desire to incorporate the arbitrage.

\end{exa}

%

Expectation $\espalt{}{0}{h}$ provides a unique arbitrage-free price in the sense of Definition \ref{D:strong_arb_free_price}, but this definition is far too restrictive, as it yields no connection between arbitrage-free prices and optimal positions. However, the range of prices which lead to finite demands (i.e.~arbitrage-free in the sense of Definition \ref{D:util_arb_free_price}) is maximally wide. One way to obtain unique price, or ``value'', for $h$ which takes into account investor preferences, is through indifference valuation. Define the (per unit, bid) indifference value $p^I$ for $\qtau$ units of $h$ through the balance equation
\begin{align*}
u(x-\qtau p^I;\Sigma_1 + \qtau h)=u(x;\Sigma_1).
\end{align*}
For exponential preferences, $p^I$ is independent of $x$, and hence we write $p^I(\qtau;\Sigma_1)$. Furthermore, from Proposition \ref{P:powerUnconstrained} we obtain the explicit formula
\begin{equation}\label{E:indiff_px_all}
p^I\left(\qtau;\Sigma_1\right) =-\frac{1}{\alpha \qtau}\log\left(\frac{u(0;\Sigma_1 + \qtau h)}{u(0;\Sigma_1)}\right)=
-\frac{1}{\beta\qtau}\log\left(\frac{\espalt{}{}{e^{-\beta(\Sigma_0 + \Sigma_1 + \qtau h)}}}{\espalt{}{}{e^{-\beta(\Sigma_0+\Sigma_1)}}}\right),
\end{equation}
Note that $p^I$  depends on the large investor's endowment and risk aversion. As such, it need not be arbitrage-free in the sense of Definitions \ref{D:strong_arb_free_price}, \ref{D:arb_free_price}, however, it is always arbitrage-free in the sense of Definition \ref{D:util_arb_free_price}.

\subsection{Optimal demands and endogenous large positions}\label{S:LargePositions}

Recall that $1/\beta = 1/\alpha + 1/\gamma$ is the combined risk tolerance of the investor and market maker. Using the results of Section \ref{S:cara_li}, together with the general theory developed in \cite{MR3678485}, we now show that optimal demands are on the order of $1/\beta$, and hence large positions arise \textit{endogenously} as either the investor's or market maker's risk aversion vanishes.  To this latter point, the two situations consistent with $\gamma\rightarrow 0$ are when the number of market makers increases (as this implies growth in the aggregate risk tolerance $1/\gamma$) and as an approximation to market maker risk-neutrality. In fact, the assumption of risk neutral market makers is common in micro-structure literature (c.f.~\cite{LiuWan16} and the references therein). Remarkably, despite the investor's price impact, large positions in $h$ endogenously arise as $\gamma\rightarrow 0$, even if the investor's risk aversion $\alpha$ remains fixed.

To heuristically see why optimal positions are on the order of $1/\beta$ as $\beta\rightarrow 0$, let $p\in (\ul{h},\ol{h})$ and $\hat{\qtau}$ be the optimal demand from \eqref{E:opt_demand}.  Keep $p,\Sigma_0,\Sigma_1$ fixed, and let $\beta\rightarrow 0$. Then, $\hat{\qtau}\beta\rightarrow 0$ implies necessarily that $p=\espalt{}{}{h}$, and $\qtau\beta\rightarrow\infty$ (respectively $-\infty$) implies $p =\ul{h}$ (resp. $\ol{h}$); a contradiction.  Thus, for all traded prices $p$ except $\espalt{}{}{h}$, it must be that $\hat{\qtau}\beta\approx \ell$ for some constant $\ell\neq 0$.

We now make the above precise, considering a sequence $\beta_n\rightarrow 0$. Let Assumption \ref{A:H4} hold, and to be consistent with \cite{MR3678485}, set
\begin{equation*}
r_n\dfn \frac{1}{\beta_n} = \frac{1}{\alpha_n} + \frac{1}{\gamma_n} \rightarrow \infty,
\end{equation*}
as the aggregate risk tolerance. From \eqref{E:indiff_px_all}, we see that along the rate $\qtau = \ell r_n, \ell\in \reals\setminus\{0\}$
\begin{equation*}
p^I\left(\ell r_n;\Sigma_1\right) = -\frac{1}{\ell}\log\left(\frac{\espalt{}{}{e^{-\ell h -\beta_n(\Sigma_0 + \Sigma_1)}}}{\espalt{}{}{e^{-\beta_n(\Sigma_0 + \Sigma_1)}}}\right).
\end{equation*}
Then, as established in \cite[Section 6.2]{MR3678485}, the dominated convergence theorem yields
\begin{align}\label{E:px_converge}
p^{\infty}(\ell)&\dfn\lim_{n\rightarrow\infty} p^{I}\left(\ell r_{n};\Sigma_1\right) =-\frac{1}{\ell}\log\left(\espalt{}{}{e^{-\ell h}}\right).
\end{align}

\begin{rem}\label{R:growing_endowment} The assumption that endowments $\Sigma_0,\Sigma_1$ are independent of $n$ can be relaxed.  For example, consider when there are $n\rightarrow\infty$ market makers with identical risk aversions $\gamma>0$ and endowments $\Sigma_0$.  Here, the representative market maker's risk aversion is $\gamma_n = \gamma/n$ and the aggregate endowment is $\Sigma^n_0 = n\Sigma_0 = (\gamma/\gamma_n)\Sigma_0$.  In other words, even as $\gamma_n\rightarrow 0$, $\gamma_n\Sigma_n = \gamma\Sigma\not\rightarrow 0$, and hence the representative market maker's risk neutrality does not necessarily correspond to zero endowment. Here, the limit in \eqref{E:px_converge} is
\begin{equation*}
\lim_{n\rightarrow\infty} p^{I}\left(\ell r_{n};\Sigma_1\right) =-\frac{1}{\ell}\log\left(\frac{\espalt{}{}{e^{-\ell h-\gamma\Sigma_0}}}{\espalt{}{}{e^{-\gamma\Sigma_0}}}\right).
\end{equation*}
More generally, our analysis can handle endowments grow at the rate $r_n$, and this assumption is reasonable for a large number of market makers.
\end{rem}

Limiting indifference values and endogenous large positions are connected using the theory developed in \cite{MR3678485}. Indeed, for $p\in (\ul{h},\ol{h})$ and $n\in\mathbb{N}$ fixed, Proposition \ref{P:util_arb_free_price} yields a unique optimal demand $\hat{\qtau}_n(p)$.  The convergence in \eqref{E:px_converge}, together with the continuity of $\ell\to p^{\infty}(\ell)$ at $\ell=0$, verifies Assumption 3.3 of \cite{MR3678485} and hence by \cite[Theorems 4.3, 4.4]{MR3678485}, as $r_n\rightarrow \infty$ it holds for all $p\neq p^\infty(0)$ that
\begin{align*}
0&<\liminf_{n\rightarrow\infty}\frac{|\hat{\qtau}_{n}(p)|}{r_{n}}\leq \limsup_{n\rightarrow\infty}\frac{|\hat{\qtau}_{n}(p)|}{r_{n}}<\infty.
\end{align*}
Thus, optimal positions become large at the rate $r_n$. In fact, \cite[Corollary 4.6]{MR3678485} improves upon this by showing that since $\ell\mapsto \ell p^{\infty}(\ell)$ is strictly convex
\begin{equation}\label{E:px_limit_exists}
\lim_{n\rightarrow\infty}\frac{\hat{\qtau}_{n}(p)}{r_{n}}=\ell\in\reals\setminus\{0\}.
\end{equation}
Recall $p^\infty(0)=\espalt{}{}{h}$ (or $p^{\infty}(0) = \espalt{}{}{he^{-\gamma\Sigma_0}}/\espalt{}{}{e^{-\gamma\Sigma_0}}$ if the endowment does not vanish as in Remark \ref{R:growing_endowment}). Therefore, whenever $p\neq p^{\infty}(0)$, the optimal demand increases in magnitude to $\infty$ exactly at the rate $r_n$. The price $p^{\infty}(0)$ is the limit of the arbitrage-free prices $\espalt{0,n}{}{h}$ in the sense of Definition \ref{D:strong_arb_free_price} (as well as Definition \ref{D:arb_free_price} for $\qtau$ fixed as $n\rightarrow\infty$), so $p\neq p^{\infty}(0)$ corresponds to the large investor obtaining a very advantageous price as the market makers approach risk neutrality. Lastly, we remark that \cite[Corollary 4.6]{MR3678485} verifies \eqref{E:px_limit_exists} for any sequence of traded prices $\cbra{p_n}_{n\in\mathbb{N}}\subset (\ul{h},\ol{h})$ provided $p_n\rightarrow p\neq p^{\infty}(0)$. There is no need for $p$ to be fixed across all $n$.

We end this section with an observation.  \cite{MR3678485} makes a general connection between endogenously rising large positions and asymptotic market completeness (c.f.~Section 6.2 therein). Since we are assuming $\K^o_t =\reals^d$, for any level of market maker risk aversion, the market with price impact is complete.  Therefore, the ``asymptotic market completeness'' of \cite{MR3678485} corresponds, not to the vanishing inability to hedge claims in the price impact model, but rather to  ``asymptotic market maker risk neutrality'' in the price impact model.  However, it is possible to associate asymptotic market maker risk neutrality with asymptotic completeness, and this is done through the lens of \textit{basis risk} models (c.f.~\cite{MR2233539}). To see this, we extend the probability space to support a $d$-dimensional Brownian motion $W$ independent of $B$, and denote by $\filt^{W,B}$ the $\prob$-augmentation of $(B,W)$'s natural filtration. Next, we define $\rho \dfn \sqrt{\alpha/(\alpha+\gamma)} \in (0,1)$ and $\bar{\rho}\dfn \sqrt{1-\rho^2}$ and consider a fictitious market with $d$ tradeable assets with price process $S$ evolving as
\begin{equation*}
\frac{dS_t}{S_t} = \rho\left(dB_t - H_t(0)dt\right) + \bar{\rho}dW_t;\qquad t\leq T.
\end{equation*}
Thus, $\rho$ measures the correlation between hedgeable and un-hedgeable shocks, and $\rho\rightarrow 1$ as $\gamma\rightarrow 0$. In this market, trading strategies $\theta$ represent the dollar positions in $S$, and the induced wealth process $X(\theta)$ has dynamics $dX_t(\theta) = \theta_t'\rho(dB_t-H_t(0)dt) + \bar{\rho}\theta_t'dW_t$, for $t\leq T$.  Now, let $\pi\in\A$ or $\pi\in\A_C$ be a strategy from Section \ref{S:cara_li} and consider $\theta = \pi / (\gamma\rho)$. By first conditioning on $\filt = \filt^B$, straight-forward calculations show that under Assumptions \ref{A: integrability} and \ref{A: large_investor_integ}
\[
\tilde{u}(0;\Sigma_1)  = \sup_{\theta\in\A/(\gamma\rho)} \espalt{}{}{-e^{-\alpha\left(X_T(\theta) + \Sigma_1\right)}};\qquad
\tilde{u}_{br,C}(0;\Sigma_{1}) = \sup_{\theta\in\A_{C}/(\gamma\rho)} \espalt{}{}{-e^{-\alpha\left(X_T(\theta) + \Sigma_1\right)}}.
\]
Therefore, we see how market maker risk aversion can directly be linked to a measure of market incompleteness, even though the original price impact model is complete.

\section{Segmented markets and partial equilibrium price and quantities}\label{S:PEPQ}

Our final section endogenizes the traded price through equilibrium arguments.  Recall that for any traded price $p\in(\ul{h},\ol{h})$, if  $p\neq \espalt{}{0}{h}$ then an arbitrage opportunity arises. Also, there is a unique $p^{\infty}(0)$ such that large positions endogenously arise whenever a) traded price differs from $p^{\infty}(0)$ and b) the representative risk aversion $\gamma$ tends to $0$. In this section, we provide justification that such prices $p$ occur in equilibrium.

We consider a \textit{segmented} market. Segmented market models are common in financial literature: for example, see \cite{RahZig09, RahZig14}, where the investors' role is played by ``arbitrageurs'', and the market makers are competitive, price-taking investors. Here, we assume there are two market makers ($A$ and $B$) with different endowments $\Sigma^A_0,\Sigma^B_0$ and tradeable security payoffs $\Psi^A$, $\Psi^B$. By ``segmented'' we mean $A$ and $B$ do not trade with one another. Similarly, there are two investors with endowments $\Sigma^A_1$, $\Sigma^B_1$, and utility functions $U_A$, $U_B$. Each investor trades with her respective ``local'' market maker. The investors also trade $h$ with one another, in the form of a bilateral OTC transaction (similarly to the trades between arbitrageurs in \cite{RahZig09, RahZig14}). The idea is that market makers in a specific product/region/market are not necessarily involved in other products/regions/markets (e.g. those with different transaction costs, regulation, geography, security contracting, etc.) and investors may wish to trade with each other, in addition to their local market makers, to offset risk.

The investors seek both a quantity $\qtau$ and price $p$ at which to trade. To determine these values, we use the notion of Partial Equilibrium Price Quantity (PEPQ) from \cite{MR2667897}, and say a pair $(p^*,\qtau^*)\in\reals^2$ is a PEPQ provided
\begin{equation*}
\qtau^* \in \underset{\qtau\in\reals}{\textrm{argmax}}\cbra{u_A(x_A - p^*\qtau,\Sigma_A + \qtau h)}\bigcap \underset{\qtau\in\reals}{\textrm{argmax}}\cbra{u_B(x_B + p^*\qtau,\Sigma_B - \qtau h)},
\end{equation*}
where $u_A$ and $u_B$ are the investors' value functions. As such, a PEPQ ensures bilateral clearing when acting optimally. For exponential utilities, the analysis simplifies, since  \eqref{E:indiff_px_all} implies $(p^*,\qtau^*)$ is a PEPQ provided
\begin{equation*}
\qtau^* \in \underset{\qtau\in\reals}{\textrm{argmax}}\cbra{\qtau p^{I,A}(\qtau,\Sigma^A_1)-p^*\qtau}\bigcap \underset{\qtau\in\reals}{\textrm{argmax}}\cbra{p^*\qtau - \qtau p^{I,B}(-\qtau,\Sigma^B_1)}.
\end{equation*}
Define the local combined risk tolerances $1/\beta^i \dfn 1/\gamma_i + 1/\alpha_i$ for $i\in\cbra{A,B}$. Using \cite[Theorem 5.8]{MR2667897}, we obtain the following:
\begin{prop}\label{P:PEPQ1}
Let Assumption \ref{A:H4} hold for both markets. If $\beta_A(\Sigma^A_0+\Sigma^A_1)-\beta_B(\Sigma^B_0+\Sigma^B_1)$ is not constant, there is a unique PEPQ  $(p^*,\qtau^*)$. In fact \[\qtau^*= \underset{\qtau\in\reals}{\argmax}\{\qtau p^{I,A}\left(\qtau;\Sigma^A_1\right)+\qtau p^{I,B}\left(-\qtau;\Sigma^B_1\right)\},\]
and
\begin{equation}\label{E:PEP}
p^*=\frac{\espalt{}{}{he^{-\beta^A(\Sigma^A_0+\Sigma^A_1+ \qtau^* h)}}}{\espalt{}{}{e^{-\beta^A(\Sigma^A_0+\Sigma^A_1 + \qtau^* h)}}}=\frac{\espalt{}{}{he^{-\beta^{B}(\Sigma^B_0+\Sigma^B_1 - \qtau^* h)}}}{\espalt{}{}{e^{-\beta^{B}(\Sigma^B_0 + \Sigma^B_1 - \qtau^* h)}}}.
\end{equation}
\end{prop}

In other words, the partial equilibrium price $p^*$ is the marginal valuation of both large investors when their endowments include the equilibrium quantity $\qtau^*$. In view of Proposition \ref{P:arb_free_range}, the equilibrium price need not be arbitrage-free since there are cases where it lies  outside of the range $[\ul{h}(\qtau^*),\ol{h}(\qtau^*)]$, for both investors. The reasoning stems from the discussion on limited arbitrages. In particular, if $h$ is negatively correlated with the endowment of investor A and positively correlated with the endowment of investor B, $\qtau^*$ will be positive, due to the hedging benefits. On the other hand, the existing market makers' inventories in both markets may induce relatively high indifference prices for $h$ in market A and relatively low prices in market B. Even if these prices are outside of the non-arbitrage range, the mutual benefits of hedging force the investor to the other direction of the transaction. The following example reinforces this point.


\begin{exa}\label{Ex:BachelierModel3}
For $i\in\cbra{A,B}$, we again consider the Bachelier model of Examples \ref{Ex:cons_bach}, \ref{Ex:BachelierModel} and \ref{Ex:BachelierModel2}. Indifference valuation \eqref{E:indiff_px_all} yields for $i\in\cbra{A,B}$
\[p^{I,i}\left(\qtau;\Sigma^i_1\right) =-\frac{\beta^i}{2}\int_0^T\left(\qtau y_t+2(f_t^i+g_t^i)\right)'y_tdt.\]
From Proposition \ref{P:PEPQ1}, the unique equilibrium quantity is
\begin{equation*}
\qtau^*=\frac{\int_0^T \left(\beta^B(f^B_t+g^B_t)-\beta^A(f^A_t+g^A_t)\right)'y_tdt}{(\beta^A+\beta^B) \int_0^Ty_t'y_tdt}.
\end{equation*}
The equilibrium price is
\begin{equation*}
p^*=-\Gamma\int_0^T \left(f^A_t+g^A_t+f^B_t+g^B_t\right)'y_t dt; \qquad \frac{1}{\Gamma}=\frac{1}{\alpha_A}+\frac{1}{\alpha_B}+\frac{1}{\gamma_A}+\frac{1}{\gamma_B},
\end{equation*}
so that $\Gamma$ is the aggregate risk tolerance. Clearly, there is no reason for $p^*$ to be the unique (for the respective markets) arbitrage-free price in the sense of Definition \ref{D:strong_arb_free_price}. In particular, if endowments $\Sigma^i_j$ for $i=A,B; j=0,1$ are sufficiently different, the equilibrium quantity is large, provided $\beta^i, i = A,B$ are both near $0$. This means neither investor exploits an arbitrage when trading with her respective market maker. This departs from segmented market models with no price impact, where investors both create and mutually exploit arbitrage opportunities: see \cite{RahZig09, RahZig14}.
\end{exa}

\subsection{Large partial equilibrium quantities}

We saw in Section \ref{S:LargePositions} that an investor (whose risk aversion is fixed) will endogenously own a large position in the contingent claim $h$ provided a) the market maker risk aversion $\gamma \approx 0$ and b) the traded price $p$ is not the unique $p^{\infty}(0)$.  Presently, we show the above conditions arise endogenously in two stylized, but representative, situations. First, when one investor has an existing large position in the claim, and second when each market has a large number of market makers.

\begin{exa}\label{Exa:one_in_lcr}\textit{Investor $B$ in the large claim regime, market makers near risk neutrality}.  Assume $\gamma^A,\gamma^B\rightarrow 0$, with $\Sigma^A_0,\Sigma^B_0$ fixed. Investor $A$ has no endowment, while investor $B$ has endowment $\ell/\gamma^B h, \ell\neq 0$, so she is in the large claim regime. The equilibrium pricing formula \eqref{E:PEP} specifies to
\begin{equation*}
p^*=\frac{\espalt{}{}{he^{-\frac{\gamma^A\alpha^A}{\gamma^A+\alpha^A}(\Sigma^A_0 + \qtau^* h)}}}{\espalt{}{}{e^{-\frac{\gamma^A\alpha^A}{\gamma^A+\alpha^A}(\Sigma^A_0 + \qtau^* h)}}}=\frac{\espalt{}{}{he^{-\frac{\gamma^B\alpha^B}{\gamma^B+\alpha^B}\left(\Sigma^B_0+\left(\frac{\ell}{\gamma^B} - \qtau^*\right)h\right)}}}{\espalt{}{}{e^{-\frac{\gamma^B\alpha^B}{\gamma^B+\alpha^B}\left(\Sigma^B_0 + \left(\frac{\ell}{\gamma^B} - \qtau^*\right) h\right)}}}.
\end{equation*}
Let $\gamma^A,\gamma^B\rightarrow 0$ and assume $\gamma^A/\gamma^B\rightarrow \delta \in (0,\infty)$. If $p^*\rightarrow p^{\infty}(0) = \espalt{}{}{h}$ then necessarily the equation for $B$ implies $\qtau^* \approx \ell/\gamma^B$.  However, by considering the equation for $A$, this in turn implies $p^* = \espalt{}{}{he^{-\ell\delta h}}/\espalt{}{}{e^{-\ell\delta h}} \neq \espalt{}{}{h}$.  Thus, it must be that $p^*$ does not limit to $p^{\infty}(0)$, and large positions endogenously arise for $A$.  In fact, it is easy to see that $\qtau^* \approx \ell /(\gamma^A+\gamma^B)$.
\end{exa}

\begin{exa}\label{Exa:create_impact}\textit{Large number of market makers.} This example emphasizes that, under the presence of price impact, the large (equilibrium) positions on a claim could arise even when the large investors' initial endowments do not contain any position on the claim. As in Remark \ref{R:growing_endowment}, we assume for each $n$ there are $n$ market makers with risk aversion $\gamma^i$ and endowment $\Sigma^i_0$ (in each local market $i\in \cbra{A,B}$). We let each large investor have fixed endowment, and in fact without loss of generality (for this example), set $\Sigma^i_1 = 0$.  Here, \eqref{E:PEP} specifies to
\begin{equation*}
p^*=\frac{\espalt{}{}{he^{-\frac{\gamma^A\alpha^A}{\gamma^A/n+\alpha^A}\left(\Sigma^A_0 + \frac{\qtau^*_n}{n} h\right)}}}{\espalt{}{}{e^{-\frac{\gamma^A\alpha^A}{\gamma^A/n+\alpha^A}\left(\Sigma^A_0 + \frac{\qtau^*_n}{n} h\right)}}}=\frac{\espalt{}{}{he^{-\frac{\gamma^B\alpha^B}{\gamma^B/n+\alpha^B}\left(\Sigma^B_0-\frac{\qtau^*_n}{n}h\right)}}}{\espalt{}{}{e^{-\frac{\gamma^B\alpha^B}{\gamma^B/n+\alpha^B}\left(\Sigma^B_0 - \frac{\qtau^*_n}{n}h\right)}}}.
\end{equation*}
Now, assume optimal positions are dominated by $n$ in that $|\qtau^*_n|/n\rightarrow 0$.  As $n\uparrow\infty$, this implies
\begin{equation*}
\frac{\espalt{}{}{he^{-\gamma^A\Sigma^A_0}}}{\espalt{}{}{e^{-\gamma^A\Sigma^A_0}}}=\frac{\espalt{}{}{he^{-\gamma^B\Sigma^B_0}}}{\espalt{}{}{e^{-\gamma^B\Sigma^B_0}}}.
\end{equation*}
Therefore, we see that except in the highly particular case when the above equality holds, large positions will ``spontaneously'' (i.e.~when neither large investor owns a position in $h$ ahead of time) arise as the number of market makers becomes large.
\end{exa}
\bigskip

\appendix

\section{Proofs}\label{AS:Proofs}

\begin{proof}[Proof of Lemma \ref{L:PI_gains_process}]

From \cite[Theorem 3.2]{MR3005018} and \cite[Theorem 4.9]{MR3375887} we know that under Assumption \ref{A: integrability} the gains process is well defined for all locally bounded predictable strategies $Q$.  It remains to obtain the explicit expression for $V$ in \eqref{E:V_good}. Using the notation of \cite{MR3005018,MR3375887}, for $(v,x,q)\in (0,\infty)\times\reals\times\reals^d$ and $(u,y,q)\in (-\infty,0)\times(0,\infty)\times\reals^d$ define $\Sigma(x,q) \dfn \Sigma_0 + x + q'\Psi$, $r(v,x) \dfn -ve^{-\gamma x}$, $F_t(v,x,q) \dfn \condespalt{}{r(v,\Sigma(x,q)}{\F_t}$ and $G_t(u,y,q) \dfn \sup_{v > 0}\inf_{x\in\reals}\left(uv + xy - F_t(v,x,q)\right)$. Direct calculation yields
\begin{equation}\label{E:1}
\begin{split}
F_t(v,x,q) &= -ve^{-\gamma x}N_t(q);\qquad G_t(u,y,q) = \frac{y}{\gamma}\log\left(\frac{N_t(q)}{-u}\right).
\end{split}
\end{equation}
Following the terminology of \cite[Equations (3.19) and (4.16)]{MR3375887}, we define the market maker's expected utility process $\{U_t\}_{t\leq T}$  by $U_{t}= U_{t}(v,x,q)=(\partial/\partial v)F_t(v,x,q)$ for $t\leq T$. From \eqref{E:Nq} and \eqref{E:1}, we readily get that  $U_{t}$ will solve the affine stochastic differential equation $dU_t/U_t = H_t(Q_t)'dB_t$ for $t\leq T$, which admits the explicit strong solution
\begin{equation*}
U_t = U_0 \mathcal{E}\left(\int_0^\cdot H_s(Q_s)'dB_s\right)_t;\qquad U_0 = \espalt{}{}{-e^{-\gamma\Sigma_0}} = -N_0(0),
\end{equation*}
provided  $Q\in\A_{PI}$.  Continuing and slightly abusing the notation, \cite[Equation (4.19)]{MR3375887} implies the gains process $V_t(Q)$ is
\begin{equation*}
\begin{split}
V_t(Q) &= -G_t(U_t,1,0) = -\frac{1}{\gamma}\log\left(\frac{N_t(0)}{-U_t(Q)}\right);\\
&= -\frac{1}{\gamma}\log\left(\frac{N_t(0)}{N_0(0)}\right) + \frac{1}{\gamma}\int_0^t H_s(Q_s)'dB_s  - \frac{1}{2\gamma}\int_0^t |H_s(Q_s)|^2ds.
\end{split}
\end{equation*}
Now, from \eqref{E:Nq} we know $N_t(0)/N_0(0) = \mathcal{E}\left(\int_0^\cdot H_s(0)'dB_s\right)_t$, which implies
\begin{equation*}
V_t(Q) = \frac{1}{\gamma}\int_0^t (H_s(Q_s)-H_s(0))'dB_s  - \frac{1}{2\gamma}\int_0^t\left( |H_s(Q_s)|^2  - |H_s(0)|^2\right)ds.
\end{equation*}
The latter expression coincides with \eqref{E:V_good}, finishing the proof.

\end{proof}


\begin{proof}[Proof of Proposition \ref{P:powerUnconstrained}]
In light of Proposition \ref{P:power} it suffices to show  both $\tilde{u}_C(0;\Sigma_1)$ and $\tilde{u}(0;\Sigma_1)$ take the value in \eqref{E:vf_uncon}. From \eqref{E:rn_deriv}, \eqref{E:tprob_def}, the density of the (unique) martingale measure $\qprob_0$ with respect to $\tprob$ on $\F_T$ is
\begin{equation}\label{E:qprob_tprob_rn}
\wt{Z}_T = \left.\frac{d\qprob_0}{d\tprob}\right\vert_{\F_T} =\left.\frac{d\qprob_0}{d\prob}\right\vert_{\F_T}\times \left.\frac{d\prob}{d\tprob}\right\vert_{\F_T} = \frac{\espalt{}{}{e^{-\alpha \Sigma_1}}}{\espalt{}{}{e^{-\gamma\Sigma_0}}}e^{\alpha\Sigma_1 - \gamma\Sigma_0}.
\end{equation}
Recall from Proposition \ref{P:power} that $p=-\alpha/\gamma$ and $q\dfn p/(p-1) = \alpha/(\alpha+\gamma)$. A straight-forward calculation shows
\begin{equation}\label{E:rn_to_q}
\wt{Z}^q_T = \left(\frac{\espalt{}{}{e^{-\alpha\Sigma_1}}}{\espalt{}{}{e^{-\gamma\Sigma_0}}}\right)^{\frac{\alpha}{\alpha+\gamma}} \frac{e^{-\beta\left(\Sigma_0 + \Sigma_1\right)}}{e^{-\alpha\Sigma_1}}.
\end{equation}
Thus, $\tespalt{}{}{\wt{Z}_T^q} < \infty$ and by the standard complete market duality theory for power utility with initial wealth $e^{\gamma \times 0} = 1$ (c.f.~\cite{MR1640352}, \cite[Lemma 5]{MR2932547}) we obtain by explicit computation
\[\tilde{u}(0,\Sigma_1) = \frac{\alpha}{\gamma}\espalt{}{}{e^{-\alpha\Sigma_1}}\left(\frac{1}{p}\tespalt{}{}{\wt{Z}_T^{q}}^{1-p}\right)=- \espalt{}{}{e^{-\gamma \Sigma_0}}^{-\frac{\alpha}{\gamma}} \times \espalt{}{}{e^{-\beta\left(\Sigma_0 + \Sigma_1\right)}}^{\frac{\alpha}{\beta}}.
\]
As for $\hat{\pi}$, the (sufficient, c.f.~\cite[Lemma 5]{MR2932547}) first order conditions for optimality are
\begin{equation}\label{E:optimal_X}
X_T(\hat{\pi})\wt{Z}_T= \wt{Z}_T^q\times \left(\tespalt{}{}{\wt{Z}_T^q}\right)^{-1}.
\end{equation}
For a given strategy $\pi$ (and $x=0$) recall the log wealth process in \eqref{E:cons_wealth}. Using \eqref{E:rn_deriv}, \eqref{E:opt_M_uncon}, \eqref{E:qprob_tprob_rn} and \eqref{E:rn_to_q}, the log of the right hand side in \eqref{E:optimal_X} reduces to
\begin{equation*}
\int_0^T \left(\hphi_t-H_t(0)\right)'(dB_t-H_t(0)dt) - \frac{1}{2}\int_0^T |\hphi_t-H_t(0)|^2dt.
\end{equation*}
From here, it is clear $\hat{\pi}  = \hphi - H(0)$. Assumption \ref{A:H2} implies $\hat{\pi}\in\A_C$, completing the proof.

\end{proof}


\begin{proof}[Proof of Proposition \ref{P:arb_free_range}]
We first prove the statement regarding Definition \ref{D:arb_free_price}. To this end, fix $\qtau > 0$, $p\in\reals$ and $Q\in\A_{PI}$. Set $W(Q) = \qtau p + V_T(Q) -\qtau h$ and note
\begin{equation*}
\begin{split}
\frac{\espalt{}{0}{e^{\gamma(\qtau h+W(Q))}}}{\espalt{}{0}{e^{\gamma\qtau h}}} &= \frac{e^{\gamma\qtau p}\espalt{}{0}{e^{\gamma V_T(Q)}}}{\espalt{}{0}{e^{\gamma\qtau h}}}\leq \frac{e^{\gamma\qtau p}}{\espalt{}{0}{e^{\gamma\qtau h}}} = e^{\gamma\qtau(p-\ol{h}(\qtau))},
\end{split}
\end{equation*}
where the inequality used \eqref{E:V_good}.  Therefore, if $p\leq \ol{h}(\qtau)$ then $p$ satisfies part $\textit{(a)}$ in Definition \ref{D:arb_free_price}. If $p>\ol{h}(\qtau)$, we observe for $Q\in\A_{PI}$ from \eqref{E:short_rep}
\begin{equation*}
W(Q) = \qtau p + V_T(Q) - \qtau h = \qtau(p - \ol{h}(\qtau)),
\end{equation*}
so  $p$ does not satisfy part $\textit{(a)}$ of Definition \ref{D:arb_free_price}. This gives the upper bound.  For the lower bound, now set $\tilde{W}(Q) = -\qtau p + V_T(Q) + \qtau h$. We then have
\begin{equation*}
\begin{split}
\frac{\espalt{}{0}{e^{\gamma(-\qtau h+\tilde{W}(Q))}}}{\espalt{}{0}{e^{-\gamma\qtau h}}} &= \frac{e^{-\gamma\qtau p}\espalt{}{0}{e^{\gamma V_T(Q)}}}{\espalt{}{0}{e^{-\gamma\qtau h}}}\leq \frac{e^{-\gamma\qtau p}}{\espalt{}{0}{e^{-\gamma\qtau h}}} = e^{-\gamma\qtau(p-\ul{h}(\qtau))}.
\end{split}
\end{equation*}
Therefore if $p\geq \ul{h}(\qtau)$, then $p$ satisfies part $\textit{(b)}$ in Definition \ref{D:arb_free_price}. If $p<\ul{h}(\qtau)$, we observe for $Q\in\A_{PI}$ from \eqref{E:short_rep} with $-h$ replacing $h$
\begin{equation*}
\tilde{W}(Q) = -\qtau p + V_T(Q) +\qtau h = \qtau(-p +\ul{h}(\qtau)),
\end{equation*}
so that $p$ does not satisfy part $\textit{(b)}$ of Definition \ref{D:arb_free_price}. This gives the lower bound, finishing the result for Definition \ref{D:arb_free_price}.

For item \textit{i)}, we note that if $p$ is arbitrage-free in the sense of Definition \ref{D:strong_arb_free_price}, then for all $\qtau>0$, $p$ is arbitrage-free in the sense of Definition \ref{D:arb_free_price}, and hence for all $\qtau>0$ we must have
\begin{equation*}
-\frac{1}{\gamma\qtau}\log\left(\espalt{}{0}{e^{-\gamma\qtau h}}\right) \leq p \leq \frac{1}{\gamma\qtau}\log\left(\espalt{}{0}{e^{\gamma\qtau h}}\right).
\end{equation*}
H\"{o}lder's inequality shows the function on the left side above is decreasing in $\qtau>0$ and the function on the right is increasing in $\qtau$.  Furthermore, by the dominated convergence theorem, the limit as $\qtau\downarrow 0$ of each of these functions is $\espalt{}{0}{h}$. Thus, $\espalt{}{0}{h}$ is arbitrage-free in the sense of Definition \ref{D:arb_free_price} for all $\qtau>0$. Also, since for all $Q\in\A_{PI}$, $\espalt{}{}{e^{\gamma V_T(Q)}} \leq 1$, we see $\espalt{}{0}{h}$ is also (vacuously) arbitrage-free at $\qtau = 0$. Finally, by the symmetry of Definition \ref{D:arb_free_price} we see $\espalt{}{0}{h}$ is also the only arbitrage-free price in the sense of Definition \ref{D:strong_arb_free_price}, even for $\qtau<0$.
\end{proof}


\begin{proof}[Proof of Proposition \ref{P:util_arb_free_price}]
Using Proposition \ref{P:powerUnconstrained} with large investor endowment $\Sigma_1+\qtau h$, we see
\begin{equation*}
u(-p\qtau;\Sigma_1 + \qtau h) = -e^{\alpha p\qtau}\espalt{}{}{e^{-\gamma\Sigma_0}}^{-\tfrac{\alpha}{\gamma}}\times\espalt{}{}{e^{-\beta\left(\Sigma_0+\Sigma_1+\qtau h\right)}}^{\tfrac{\alpha}{\beta}}.
\end{equation*}
Thus, the optimization problem is to minimize
\begin{equation*}
p\qtau + \frac{1}{\beta}\log\left(\espalt{}{}{e^{-\beta(\Sigma_0+\Sigma_1 +\qtau h)}}\right),
\end{equation*}
over $\qtau\in\reals$.  Under the given integrability assumptions the above function is both smooth and strictly convex in $\qtau$.  The first order conditions for optimality are \eqref{E:opt_demand}, and that the right hand side of \eqref{E:opt_demand} goes to $\ul{h}$ (resp.~$\ol{h}$) as $\qtau$ goes to $\infty$ (resp.~$-\infty$) has been shown in \cite{MR3668156}.
\end{proof}

\bibliographystyle{siam}


\def\polhk#1{\setbox0=\hbox{#1}{\ooalign{\hidewidth
  \lower1.5ex\hbox{`}\hidewidth\crcr\unhbox0}}}

\end{document}